\documentclass[11pt,letterpaper]{article}

\newcommand*{\myproofname}{My proof}
\newenvironment{myproof}[1][\myproofname]{\begin{proof}[#1]}{\end{proof}}

\usepackage{booktabs} 

\usepackage[square,numbers]{natbib}
\usepackage{amsmath,amsthm,amssymb}
\usepackage[margin=1in]{geometry}
\usepackage{enumitem}
\usepackage{xfrac}
\usepackage{subcaption}
\usepackage{hyperref}

\usepackage[ruled]{algorithm2e} 

\SetAlFnt{\small}
\SetAlCapFnt{\small}
\SetAlCapNameFnt{\small}
\SetAlCapHSkip{0pt}
\IncMargin{-\parindent}

\usepackage{tikz}

\newtheoremstyle{theoremc}
{4pt}
{4pt}
{\itshape}
{0pt}
{\bfseries}
{.}
{ }
{\thmname{#1}\thmnumber{ #2}\thmnote{ (#3)}}
\theoremstyle{theoremc}

\newtheorem{observation}{Observation}

\newtheorem{claim}{Claim}
\newtheorem{example}{Example}
\newtheorem{theorem}{Theorem}
\newtheorem*{theorem*}{Theorem}
\newtheorem{definition}{Definition}
\newtheorem{proposition}{Proposition}
\newtheorem{corollary}{Corollary}
\newtheorem{lemma}{Lemma}

\usepackage{titlesec}
\titlespacing{\paragraph}{%
	0pt}{
	0.2\baselineskip}{
	1em}
\titlespacing*{\section}
{0pt}{3ex plus 1ex minus .2ex}{1.5ex plus .2ex}
\titlespacing*{\subsection}
{0pt}{1ex plus 1ex minus .2ex}{0.5ex plus .2ex}

\usepackage{mathptmx} 
\DeclareMathAlphabet\mathcal{OMS}{cmsy}{m}{n}
\SetMathAlphabet\mathcal{bold}{OMS}{cmsy}{b}{n}

\newcommand{\pdnote}[1]{{\color{black}{#1}}}

\newcommand{\itnote}[1]{{\color{black}{#1}}}

\begin{document}

\title{Simple versus Optimal Contracts}

\author{Paul D\"utting\thanks{Department of Mathematics, London School of Economics, Houghton Street, London WC2A 2AE, UK. Email: {\tt p.d.duetting@lse.ac.uk}.} \and Tim Roughgarden\thanks{Department of Computer Science, Columbia University, 500 West 120 Street, New York, NY 10027. Email: {\tt tr@cs.columbia.edu}.} \and Inbal Talgam-Cohen\thanks{Computer Science Department, Technion -- Israel Institute of Technology, Haifa, Israel 3200003. Email: {\tt italgam@cs.technion.ac.il}.}}





\date{}

\maketitle

\begin{abstract}
We consider the classic principal-agent model of contract theory, in which a principal designs an outcome-dependent compensation scheme to incentivize an agent to take a costly and unobservable action. When all of the model parameters---including the full distribution over principal rewards resulting from each agent action---are known to the designer, an optimal contract can in principle be computed by linear programming. In addition to their demanding informational requirements, however, such optimal contracts are often complex and unintuitive, and do not resemble contracts used in practice. 

This paper examines contract theory through the theoretical computer science lens, with the goal of developing novel theory to explain and justify the prevalence of relatively simple contracts, such as linear (pure commission) contracts.  First, we consider the case where the principal knows only the first moment of each action's reward distribution, and we prove that linear contracts are guaranteed to be worst-case optimal, ranging over all reward distributions consistent with the given moments.  Second, we study linear contracts from a worst-case approximation perspective, and prove several tight parameterized approximation bounds.
\end{abstract}




\section{Introduction}
\label{sec:intro}
\paragraph{Classic contract theory.}
Many fundamental economic interactions can be phrased in terms of two
parties, a \emph{principal} and an \emph{agent}, where the agent
chooses an action and this imposes some (negative or positive)
externality on the principal. Naturally, the principal will want to
influence which action the agent chooses. This influence will often
take the form of a \emph{contract}, in which the principal compensates
the agent contingent on either the actions or their outcomes; with the
more challenging and realistic scenario being the one in which the
principal cannot directly observe the agent's chosen action. Instead,
the principal can only observe a stochastic outcome that results from
the agent's action.

For example, consider a salesperson (agent) working for a company (principal) producing a
range of products with different revenue margins. The salesperson
chooses the amount of effort to spend on promoting the various products.
The company may not be able to directly observe effort, but can
presumably track the number of orders the salesperson generates.
Assuming this number is correlated with the salesperson's actions (the
harder he works, the more sales of higher margin products he
generates),%
\footnote{For clarity we shall address the agent as male and the principal as female.} 
it may make sense for the company to base his pay on
sales---i.e., to put him on commission---to induce him to expend the
appropriate level of effort.
Another example is the interaction between a car owner (agent) and an insurance 
company (principal). The car owner's behavior influences the risks his car is exposed 
to. If realized risks are borne entirely by the insurance company, the owner might not 
bother to maintain the asset properly (e.g., might park in a bad neighborhood).  
This is an example of the well-known {\em moral hazard} problem.  
Typically, such bad behavior is difficult to contract against directly.  
But because this behavior imposes an externality on the insurance company, 
companies have an interest in designing contracts that guard against~it.

Both of these examples fall under the umbrella of the \emph{hidden action 
principal-agent model}, arguably the central model of \emph{contract theory}, 
which in turn is an important and well-developed area within microeconomics.%
\footnote{For example, the 2016 Nobel Prize in economics was awarded to 
Oliver Hart and Bengt Holmstr\"om for their contributions to contract 
theory. The prize announcement stresses the importance of contracts in practice, describing modern economies as ``held together 
by innumerable contracts'', and explains how the theoretical tools created by Hart and Holmstr\"om are invaluable to their understanding.}
Perhaps surprisingly, this area has received far less attention from the 
theoretical computer science %
community than auction and mechanism design, despite its strong ties to 
optimization, and potential applications ranging from crowdsourcing 
platforms~\citep{HoSV16}, through blockchain-based smart 
contracts~\citep{CongH18}, to incentivizing quality healthcare~\citep{BastaniBBGJ18}.

\paragraph{The model.}
Every instance of the model can be described by a
pair $(A_n,\Omega_m)$ of $n$ \emph{actions} and $m$ \emph{outcomes}. In the 
salesperson example, the actions are the levels of effort and the outcomes 
are the revenues from ordered products.
As in this example, we usually identify the (abstract) outcomes with the (numerical) 
rewards to the principal associated with them. 
The agent chooses an action $a_i\in A_n$, unobservable to the principal, which 
incurs a \emph{cost} $c_i\ge 0$ for the agent, and results
in a distribution $F_i$ with expectation $R_i$ over the outcomes in
$\Omega_m$.
The realized outcome $x_j\ge 0$ is awarded to the principal. 

The principal designs a contract that specifies a \emph{payment} $t(x_j)\ge 0$ to the 
agent for every outcome~$x_j$ (since the outcomes, unlike the actions, are 
observable to the principal). This induces an expected payment 
$T_i=\mathbb{E}_{x_j\sim F_i}[t(x_j)]$ for every action $a_i$. The agent then chooses 
the action that maximizes his expected \emph{utility} $T_i-c_i$ over all actions 
(``incentive compatibility''), or opts out of the contract if no action with non-negative 
expected utility exists (``individual rationality''). 

As the design goal, the principal wishes to maximize her expected 
\emph{payoff}: the expected outcome $R_i$ minus the agent's expected 
payment $T_i$, where $a_i$ is the action chosen by the agent. Therefore 
contract design immediately translates into the following \emph{optimization 
problem}: given $(A_n,\Omega_m)$, find 
a payment vector $t$ that maximizes $R_i-T_i$, where $a_i$ is incentive 
compatible and individually rational for the agent. 
We focus on the \emph{limited liability} case, where the contract's
payments $t$ are constrained to be non-negative (i.e., money only 
flows from the principal to the agent).%
\footnote{Without some such assumption there is a simple but 
unsatisfying optimal solution for the principal when the agent 
is risk-neutral: simply sell the project to the agent, at a price 
just below the maximum expected welfare $R_i-c_i$ the agent 
can generate by choosing an action. The agent may as well accept 
(and then select the welfare-maximizing action), and the 
principal pockets essentially the full welfare. This solution is 
incompatible with practical principal-agent settings, e.g., a 
salesperson does not typically buy the company from its owner.}
A detailed description of the model appears in Section \ref{sec:model}.

\paragraph{Optimal contracts and their issues.} \itnote{With perfect knowledge of the distributions mapping actions to outcomes,} it is straightforward
to solve the optimization problem associated with finding the optimal
contract for the principal, by solving one
linear program per action. Each linear program minimizes the expected
payment to the agent subject to the constraint that he prefers this
action to any other---including opting out of the contract---and
subject to the payments all being non-negative. The best of these
linear programs (achieving the highest expected payoff for the
principal) gives the best action to incentivize and an optimal
contract incentivizing it.  

In addition to their demanding informational requirements, however,
such optimal contracts are often complex and unintuitive, and do not 
resemble contracts used in practice.\footnote{A similar issue arises in 
auction theory: linear programs can be used to characterize optimal 
auctions, which often turn out to be impractically complicated and 
unintuitive (see, e.g., \cite{Hartline17}).} Example~\ref{ex:bad-example} 
demonstrates this critique---it is not clear how to interpret the optimal 
contract nor how to justify it to a non-expert. The optimal payment 
scheme in this example is not even \emph{monotone}, i.e., a better 
outcome for the principal can result in a lower payment for the agent! 
In the salesperson example, this would create an incentive for the 
salesperson to manipulate the outcome, for example by hiding or canceling orders.

\begin{example} 
	\label{ex:bad-example}
	\pdnote{There} are $m = 6$ outcomes
        $x = (1, 1.1, 4.9, 5, 5.1, 5.2)$, 
        and $n = 4$ actions
        with the following outcome
        distributions and costs:
        $F_1 = (\sfrac{3}{8}, \sfrac{3}{8}, \sfrac{1}{4}, 0, 0, 0)$,
        $F_2 = (0, \sfrac{3}{8}, \sfrac{3}{8}, \sfrac{1}{4}, 0, 0)$,         $F_3 = (0, 0, \sfrac{3}{8}, \sfrac{3}{8}, \sfrac{1}{4}, 0)$,
        $F_4 = (0, 0, 0, \sfrac{3}{8}, \sfrac{3}{8}, \sfrac{1}{4})$,
        and $(c_1,c_2,c_3,c_4) = (0,1,2,2.2)$.
The LP-based approach shows that the optimal contract in this case incentivizes
action $a_3$ with payments 
$$
t \approx (0, 0, 0.15, 3.93, 2.04, 0).
$$ 
The analysis appears 
for completeness in Appendix~\ref{appx:example-analysis}.
\end{example}

\paragraph{Linear contracts as an alternative.}
Perhaps the simplest non-trivial contracts are \emph{linear}
contracts, where the principal commits to paying the
agent an $\alpha$-fraction of the realized outcome (i.e., 
payments are linear in the outcomes).  Unlike optimal contracts, linear
contracts are conceptually simple and
easy to explain; their payments are automatically monotone; and they are the most ubiquitous contract form in practice.%
\footnote{\label{ftnt:debt}To our knowledge, the only other common 
	contract form according to the economics literature is ``debt contracts,'' 
	which are similar to linear contracts except with zero payments 
	for a set of the lowest outcomes \cite{Hebert17}. 
	Our results do not qualitatively change for such contracts---see 
	Section \ref{sec:beyond-linear}.}
From an optimization standpoint, however, they can be suboptimal even in
simple settings, as the next example demonstrates:

\begin{example}
\pdnote{There} are $m = 2$ outcomes $x = (1,3)$, and $n = 2$ actions $a_1$ 
and $a_2$ with $F_1 = (1,0), c_1 = 0$ and $F_2 = (0,1), c_2 = \sfrac{4}{3}$, 
respectively.
The optimal contract incentivizes action $a_2$ with payments $t=(0,\sfrac{4}{3})$, 
resulting in expected payoff of $3-\sfrac{4}{3} = \sfrac{5}{3}$ for the principal. 
The maximum expected payoff
of any linear contract is~$1$ (regardless of which action is
incentivized).
\end{example}

\paragraph{Simple versus optimal contracts in the economic literature.}

The critique \itnote{that optimal contracts are unrealistically complex} is well known to economists. Already in 1987 Milgrom and Holmstr\"om \cite[p.~326]{HolmstromMilgrom87} wrote: %
\begin{quote}
``It is probably the great robustness of linear rules based on aggregates that accounts
for their popularity. That point is not made as effectively as we would like by our model;
we suspect that it cannot be made effectively in any traditional Bayesian model.''
\end{quote}
Indeed, for several decades the economic literature struggled to find the right tools to explain
the prevalence of simple, in particular, linear contracts in practice.

\itnote{A breakthrough was the recent paper by \citet{Carroll15}, in which he 
proposes ``a forthrightly non-Bayesian model of robustness'' (p.~537).
In this work, a key change to the standard principal-agent model is introduced: the set of actions 
available to the agent is \emph{completely unknown} to the principal. 
Because in this new setting no guarantee is possible, Carroll relaxes the new model by adding the assumption that \emph{some} set of actions $\mathcal A_0$ is fully known to the principal (that is, she is fully aware of the distributions the actions in $\mathcal A_0$ induce over outcomes as well as of their costs). Carroll discusses two possible ways in which to ``make sense of this combination of non-quantifiable and quantifiable uncertainty'' (i.e., completely unknown 
and known actions---see p.~546). The main result is that a linear contract is max-min optimal 
in the worst case over all possible sets of actions $\mathcal A$ that contain the known set $\mathcal A_0$ (where 
$\mathcal A\setminus \mathcal A_0$ can be anything). Carroll sees the main contribution not necessarily in a literal interpretation 
of the model, but rather in finding ``a formalization of a robustness property of linear 
contracts---a way in which one can make guarantees about the principal's payoff with very 
little information'' \cite[Sec.~2.1]{Carroll18}.}

\subsection{Our Contributions}

Our main goal is \emph{to initiate the study of simple contracts and their guarantees through 
the lens of theoretical computer science}.  We utilize tools and ideas from worst-case 
algorithmic analysis and
prior-independent auction design to make contributions in two main
directions, both justifying and informing the use of linear contracts.

\paragraph{Max-min robustness of linear contracts.}  Our first
contribution is a new notion of robustness for linear
contracts. 
Instead of assuming that there are completely unknown actions
available to the agent alongside a fully known action set, we assume
that the principal has partial knowledge of all actions---she knows
their costs and has \emph{moment information} about their
reward distributions. 
Assuming moment information is natural from an optimization perspective (the idea can at least be traced back to Scarf's seminal 1958 paper on distributionally-robust stochastic programming \cite{Scarf58}), and it's a standard assumption in prior-independent auction design (see, e.g., \cite{AzarDMW13,BandiBertsimas14,CarrascoEtAl17}).
The assumption is particularly natural in the context of contracts and the principal-agent model, where it is the minimal assumption that allows for optimization over linear contracts. With this assumption, the optimal linear contract is a natural contract to focus on, and our analysis shows that it is in fact \emph{max-min optimal} across all possible contracts (Theorem~\ref{thm:robustness} in Section~\ref{sec:robustness}).\footnote{Our max-min optimality result is conceptually related to work by Carroll~\cite{Carroll17}, Gravin and Lu~\cite{GravinL18}, and Bei et al.~\cite{BeiGLT19}. These papers show the max-min optimality of selling items separately when the marginals of all items are held fixed but the joint distribution is unknown (i.e., adversarial). Knowing the marginals is precisely what is needed for optimizing over mechanisms that sell each item separately.}
Our result thus offers an alternative
formulation of the robustness property of linear contracts, in a
natural model of moment information that is easy to interpret.

\begin{theorem*}[See Section \ref{sec:robustness}]
For every outcome set, action set, action costs, and expected action rewards, 
a linear contract maximizes the principal's worst-case expected payoff, 
where the worst case is over all reward distributions with the given expectations.
\end{theorem*}

\paragraph{Approximation guarantees.}
Our second contribution is to conduct \emph{the first study of simple 
contracts from an approximation perspective}. 
Studying the worst-case approximation guarantees of classic
microeconomic mechanisms---linear contracts in this case---has been
a fruitful approach in other areas of algorithmic game theory. 
Applying this approach, we achieve a \emph{complete and tight} analysis 
of the approximation landscape for linear contracts.
Our analysis %
shows that \emph{linear contracts 
are approximately optimal except in pathological cases, which have
many actions, a
big spread of expected rewards among actions, and a
big spread of costs among actions, simultaneously.} Concretely, in the pathological cases we construct for our lower bounds (see e.g.~Theorem \ref{thm:lower-bound} in Section \ref{sec:linear-contracts}), the expected rewards of actions $a_1,a_2,a_3,\dots$ are $1,\sfrac{1}{\epsilon}, \sfrac{1}{\epsilon^2},\dots$ for vanishing $\epsilon$, and the costs are likewise ``exploding'' and tailored such that the difference between the exponentially large expected reward and exponentially large cost of the $i$th action is roughly $i$. The following summarizes our findings more formally:

\begin{theorem*}[See Section \ref{sec:linear-contracts}]
Let $\rho$ denote the worst-case ratio between the expected principal payoff 
under an optimal contract and under the best linear contract. Then
\begin{itemize}[topsep=0ex,itemsep=0ex,parsep=0ex]
\item [(a)] Among principal-agent settings with $n$ actions, $\rho=n$.

\item [(b)] Among settings where the ratio of the highest and lowest 
expected rewards is~$H$,~$\rho=\Theta(\log H)$.

\item [(c)] Among settings where the ratio of the highest and 
lowest action costs is~$C$, $\rho=\Theta(\log C)$.

\item [(d)] Among settings with $m \ge 3$ outcomes, $\rho$ can 
be arbitrarily large in the worst case.
\end{itemize}
\end{theorem*}

The upper bounds summarized in the above theorem hold even with
respect to the strongest-possible benchmark of the optimal expected welfare
(rather than merely the optimal principal expected payoff); they thus
answer the natural question of how much of the ``first-best'' welfare
a linear contract can extract.  The matching lower bounds in the above
theorem all apply even when we add a strong regularity condition to
the principal-agent settings, known in the literature as the \emph{monotone likelihood ratio
property} (MLRP) (see Appendix \ref{appx:MLRP}).
In Section \ref{sec:beyond-linear}, we show an extension 
of our lower bounds 
to all \emph{monotone} (not necessarily linear) contracts: we prove 
that among principal-agent settings with $n$ actions, the ratio 
between the expected principal payoff under an optimal contract and 
under the best \emph{monotone} contract can be $n-1$ in the (pathological) worst case.

\paragraph{Discussion.} 
We view parts~(a)--(c) of the above theorem on approximation guarantees as
surprisingly positive results. A priori, it is completely unclear
which of the many parameters of principal-agent settings, if any, governs 
the performance of simple contracts. Our results show that there is 
no finite approximation bound that holds uniformly over {\em all} of 
the model parameters, but that we can obtain the next best thing: by 
fixing \emph{just one} of the model's ingredients (either the number of 
actions, the range of the outcomes, or the range of the costs, as 
preferred), it is possible to obtain an approximation guarantee that 
holds uniformly over all other parameters.  Our theorem shows 
that linear contracts are far from optimal only when the number 
of actions is large, \emph{and} there is a huge spread in expected 
rewards, \emph{and} there is a huge spread of action costs.  
Few if any of the most popular instantiations of the principal-agent 
model have all three of these properties.

\paragraph{Organization.}
Section \ref{sec:model} introduces the model.
Section \ref{sec:geometry} establishes several basic but fundamental properties of linear contracts. Sections \ref{sec:robustness} and \ref{sec:linear-contracts} contain our main results on their performance, and Section \ref{sec:beyond-linear} goes beyond linear to monotone contracts. Section \ref{sec:conclusion} concludes.

\subsection{Further Related Work}

Contract theory is one of the pillars of microeconomic theory. We refer 
the interested reader to the classic papers of 
\citep{Shavell79,GrossmanHart83,Rogerson85,HolmstromMilgrom91}, the coverage in \cite[Chapters 13-14]{MasColell95}, the excellent introductions of \cite{CaillaudHermalin00,TadelisSegal05}, and the comprehensive textbooks \cite{LaffontMartimort02,BoltonDewatripont04}. 

\paragraph{Computational approaches to contract design.}

To our knowledge, decidedly computational approaches to contract design 
have appeared so far only in the work of \cite{Babaioff-et-al12} 
(and follow-ups \cite{Babaioff-et-al09,Babaioff-et-al10}), the work of 
\cite{BabaioffWinter14}, and the work of \cite{HoSV16} (and follow-up \cite{KorenC18}). 
The first paper \cite{Babaioff-et-al12} initiates the study of a related but 
different model known as \emph{combinatorial agency}, in which combinations 
of agents replace the single agent in the classic principal-agent model. 
The challenge in the new model stems from the need to incentivize multiple agents, 
while the action structure of each agent is kept simple (effort/no effort). The focus of this 
line of work is on complex combinations of agents' efforts influencing the outcomes, 
and how these determine the subsets of agents to contract with.
The second paper \cite{BabaioffWinter14} introduces a notion of contract 
complexity based on the number of different payments specified in the contract, and 
studies this complexity measure in an $n$-player normal-form game framework. 
In their framework there are no hidden actions, making our model very different from theirs.
\itnote{
The third paper \cite{HoSV16} develops a model of \emph{dynamic} contract 
design: in each sequential round, the principal determines a contract, an agent 
arrives and chooses an action (effort level), and the principal receives 
a reward. Agents are drawn from an unknown prior distribution 
that dictates their available actions. The problem thus reduces to a multi-armed 
bandit variant with each arm representing a potential contract. The main focus 
of this line of work is on implicitly learning the underlying agent distribution 
to minimize the principal's regret over time. 
}

\paragraph{(Non)relation to signaling.}
Since one of the main features of the principal-agent model is the
information asymmetry regarding the chosen action (the agent knows
while the principal is oblivious), and due to the ``principal'' and
``agent'' terminology, on a superficial level contract theory may seem
closely related to signaling
\cite{Spence73,MaskinTirole90,MaskinTirole92}. 
This is not the case, and the relationship is no closer than that
between auction theory and signaling. As \citet{Dughmi17}
explains, the heart of signaling is in creating the right information
structure, whereas the heart of contract design is in setting the
right payment scheme.%
\footnote{``There are two primary ways of influencing the behavior of
  self-interested agents: by providing incentives, or by influencing
  beliefs. The former is the domain of traditional mechanism design,
  and involves the promise of tangible rewards such as [...] money.
  The latter [...] involves the selective provision of payoff-relevant
  information to agents through strategic communication''
  \cite[p.~1]{Dughmi17}.}
Put differently, in signaling, it is the more-informed party that
faces an economic design problem;
in hidden-action contract theory, it
is the less-informed party (i.e., the principal).
For more on signaling from a computational perspective the reader 
is referred to \cite{Emek-et-al-14,MiltersenSheffet12,Dughmi-et-al14}.

\itnote{
\paragraph{Concurrent work on principal-agent settings with no money.} Two recent papers \cite{KhodabakhshPT18,KleinbergK18}} study algorithmic aspects of another loosely related but distinct problem called \emph{optimal delegation} \cite{ArmstrongV10}. In this problem, a principal has to search for and decide upon a solution, and wishes to delegate the search to an agent with misaligned incentives regarding which solution to choose.
Crucially, there are \emph{no monetary transfers}, making the problem quite different from ours. A third recent paper \cite{KR18} studies a principal and agent setting where the agent has a ``budget'' of effort to spread across different actions (rather than a cost per level of effort), and any effort spent within the budget does not lower the agent's utility. There are again no monetary transfers from the principal to the agent, and the principal incentivizes the agent by mapping his outcomes to academic grades or other classifications. 

\section{The Hidden Action Principal-Agent Model}
\label{sec:model}

\paragraph{The model.}
An instance of the principal-agent model is described by a pair 
$(A_n,\Omega_m)$ of $n$ \emph{actions} and $m$ \emph{outcomes}.
	We identify the $j$th \emph{outcome} for every $j\in [m]$ with 
	its \emph{reward}~$x_j$ to the principal, and assume w.l.o.g.~that 
	the outcomes are increasing, i.e., $0 \leq x_1 \leq x_{2} \leq \dots \leq x_m$.

	Each \emph{action} is a pair $a = (F_a,c_a)$, in which $F_a$ is 
	a distribution over the $m$ outcomes and $c_a \geq 0$ is a \emph{cost}. 
	Denote by $F_{a,j}$ the probability of action $a$ to lead to outcome $x_j$; 
	we assume w.l.o.g.~that every outcome has some action leading to it 
	with positive probability. 
	The agent chooses an action $a \in A_n$ and bears the cost $c_a$, whereas 
	the principal receives a random reward $x_j$ drawn from $F_a$. 
	Crucially, \emph{the chosen action is hidden}: the principal observes the outcome $x_j$ but not the action $a$. 

Denote the \emph{expected outcome} (i.e., reward to the principal) from 
action $a$ by $R_a = \mathbb{E}_{j\sim F_a}[x_j] = \sum_{j \in [m]} F_{a,j}x_j.$ 
The difference $R_{a}-c_{a}$ is the expected \emph{welfare} from choosing action $a$.
When an action $a_i$ is indexed by $i$, we write for brevity $R_i,F_{i,j},c_i$ (rather than $R_{a_i},F_{a_i,j},c_{a_i}$).

\paragraph{Standard assumptions.} 
Unless stated otherwise we assume:
\begin{enumerate}%
	\item[\bf A1] There are no ``dominated'' actions, i.e., every two actions $a,a'$ 
	have distinct expected outcomes $R_a\ne R_{a'}$, and the action with the higher 
	expected outcome $R_a>R_{a'}$ also has higher cost $c_a> c_{a'}$.
	\item[\bf A2] There is a unique action $a$ with maximum welfare $R_a-c_a$. 
	\item[\bf A3] There is a zero-cost action $a$ with $c_a=0$. 
\end{enumerate}

Assumption A1 means there is no action with lower expected outcome and higher cost 
than some other action, although we emphasize that there \emph{can} be 
an action with lower \emph{welfare} and higher cost (in fact, incentivizing 
the agent to avoid such actions is a source of contract complexity). Our 
main results in Sections~\ref{sec:robustness}-\ref{sec:linear-contracts} do not 
require this assumption (see Section \ref{sec:linear-contracts} for details). 
Assumptions A2 and A3 are for the sake of expositional simplicity. In particular, 
Assumption A3 means we can assume the agent does not reject a contract with 
non-negative payments, since there is always an individually rational choice of action; 
alternatively, individual rationality could have been imposed directly.

\paragraph{Contracts.} 
A contract defines a \emph{payment scheme} $t$ with a payment (transfer) 
$t_j \geq 0$ from the principal to the agent for every
outcome~$x_j$. We denote by $T_a$ the expected payment
$\mathbb{E}_{j\sim F_a}[t_j] = \sum_j F_{a,j}t_j$ for action~$a$, and
by $T_i$ the expected payment for $a_i$. Note that the payments are contingent
only on the outcomes as the actions are not observable to the
principal. The requirement that $t_j$ is non-negative for every $j$ is referred to in
the literature as \emph{limited liability} \citep{Carroll15}, and it
plays the same role as the standard \emph{risk averseness} assumption in 
ruling out trivial solutions where a contract is not actually required
\citep{GrossmanHart83}. Limited liability (or its parallel agent risk averseness) 
is the second crucial feature of the classic principal-agent model, in addition 
to the actions being hidden from the principal.

\paragraph{Implementable actions.} 

The agent's expected \emph{utility} from action $a$ given payment scheme $t$ is $T_a-c_a$. 
The agent chooses an action that is: 
 (i) \emph{incentive compatible (IC)}, 
i.e., maximizes his expected utility among all actions in $A_n$; and 
(ii) \emph{individually rational (IR)}, i.e., has non-negative expected 
utility (if there is no IR action the agent refuses the contract). 
We adopt the standard assumption that the agent tie-breaks among 
IC, IR actions in favor of the principal.%
\footnote{The idea is that one could perturb the payment schedule slightly 
to make the desired action uniquely optimal for the agent.
For further discussion see \cite[p. 8]{CaillaudHermalin00}.} 
We say a contract \emph{implements} action $a^*$ %
if given its payment scheme $t$,  the agent chooses $a^*$;
if there exists such a contract we say $a^*$ is  \emph{implementable}.

\paragraph{Optimal contracts and LPs.} 

The principal seeks an \emph{optimal} contract: a payment scheme~$t$ 
that maximizes her expected \emph{payoff} $R_{a} - T_a$,
where $a$ is the action implemented by the contract (i.e., $a$ is both IC 
and IR for the agent, with ties broken to maximize the expected payoff of the principal). 
Notice that summing up the agent's expected utility $T_a-c_a$ with the 
principal's expected payoff $R_a-T_a$ results in the contract's expected 
welfare $R_{a}-c_{a}$. A contract's payment scheme thus determines both 
the size of the pie (expected welfare), and how it is split between the principal and the agent.

An optimal contract can be found through linear programming: For each action determine the minimum expected payment at which this action can be implemented. The best of these linear programs (achieving the highest expected payoff  for the principal) gives the best action to incentivize and an optimal contract incentivizing it.

The LP for incentivizing action $a$ at minimum expected payment has $m$ payment
variables $\{t_j\}$, which by limited liability must be non-negative,
and $n-1$ IC constraints ensuring that the agent's expected utility
from action $a$ is at least his expected utility from any other
action.  Note that by Assumption~A3, there is no need for an IR
constraint to ensure that the expected utility is non-negative. The LP
is:
\begin{eqnarray*}
\min & \sum_{j\in[m]} {F_{a,j}t_j} & \label{LP:min-pay}\\
\text{s.t.} & \sum_{j\in[m]} {F_{a,j}t_j} - c_a \ge \sum_{j\in[m]} {F_{a',j}t_j} - c_{a'} & \forall a'\ne a, a'\in A_n, \\
& t_j \ge 0 & \forall j\in[m].
\end{eqnarray*}

Fairly little is known about the \emph{structure} of the optimal contracts that come out of this approach. There are a few notable exemptions from this general rule---e.g., settings with just $n = 2$ actions---which we discuss in more detail in 
Appendix~\ref{appx:simple-and-optimal}.

This LP and close variants together with their duals can also be used to characterize implementable actions, where the implementing contract is either arbitrary or required to be \emph{monotone} as defined below. The characterizations hinge on whether the action's distribution can be emulated or dominated by a combination of other actions at lower cost%
---see Appendix \ref{sub:implement} for details.
It also follows that the optimal contract will have at most
$(n-1)$ nonnegative payments---see Appendix~\ref{sub:nonzero-payments}.

\paragraph{Linear/monotone contracts.} %
	In a \emph{linear} contract the payment scheme is a linear function of the outcomes, i.e., $t_j = \alpha x_j \ge 0$ for every $j\in[m]$. We refer to $\alpha$ as the linear contract's \emph{parameter}, which is $\ge 0$ due to limited liability.
A natural generalization is a \emph{degree-$d$ polynomial} contract, in which the payment scheme is a non-negative degree-$d$ polynomial function of the outcomes: $t_j = \sum_{k=0}^d {\alpha_k x_j^k}\ge 0$ for every $j\in[m]$. If $d=1$ we get an \emph{affine} contract; such contracts play a role in Section \ref{sec:robustness}. 
Linear and affine contracts are \emph{monotone},
	where a contract is monotone if its payments are nondecreasing in the outcomes, i.e., $t_j\le t_{j'}$ for $j<j'$. 

\paragraph{Regularity assumptions.}  

The economic literature (see, e.g., \cite{GrossmanHart83}) introduces
a regularity assumption called the \emph{monotone likelihood ratio
	property} (MLRP) for principal-agent settings. Intuitively, the
assumption asserts that the higher the outcome, the more likely it is
to be produced by a high-cost action than a low-cost one (in a
relative sense).  We
adapt the standard definition to accommodate for zero probabilities,
as follows:

\begin{definition}[MLR]
	\label{def:increasing-LR}
	Let $F,G$ be two distributions over $m$ values $v_1,\dots,v_m$. 
	The likelihood ratio $F_j/G_j$ 
	is monotonically increasing in $j$ if 
	$$
	F_j/G_j \le F_{j'}/G_{j'}
	$$
	for every $j<j'$ such that at least one of $F_j,G_j$ is
	positive, and at least one of $F_{j'},G_{j'}$ is positive.
\end{definition}

\begin{definition}[MLRP]
	\label{def:MLRP}
	A principal-agent problem satisfies {\em MLRP} if for every pair of actions $a,a'$ such that $c_a<c_{a'}$, the likelihood ratio $F_{a',j}/F_{a,j}$ is monotonically increasing in $j$. 
\end{definition}

\begin{proposition}[MLR $\implies$ FOSD \cite{TadelisSegal05}]
	If the likelihood ratio $F_j/G_j$ is monotonically increasing in $j$, then $F$ first-order stochastically dominates $G$. The converse does not hold.
\end{proposition}

\paragraph{Max-min evaluation and approximation.}

We apply two approaches to evaluate the performance of simple contracts: max-min in Section \ref{sec:robustness}, and approximation in Section \ref{sec:linear-contracts}. We now present the necessary definitions, starting with the max-min approach.

	A \emph{distribution-ambiguous} action is a pair $a=(R_a,c_a)$, in which $R_a\ge 0$ is the action's expected outcome and $c_a\ge 0$ is its cost. Distribution $F_a$ over outcomes $\{x\}$ is \emph{compatible} with distribution-ambiguous action $a$ if $\mathbb{E}_{x\sim F_a}[x]=R_a$.
	A principal-agent setting $(A_n,\Omega_m)$ is \emph{ambiguous} if it has $m\ge 3$ outcomes and $n$ distribution-ambiguous actions, and there exist distributions $F_1,\dots,F_n$ over the outcomes compatible with the actions.%
	\footnote{A setting with $m=2$ outcomes cannot be ambiguous since the expectation determines the distribution; moreover the conundrum of ``optimal but complex'' vs.~``suboptimal but ubiquitous'' never arises as the optimal contract has a simple form---see Appendix \ref{appx:simple-and-optimal}.}
In ambiguous settings, it is appropriate to apply a worst-case performance measure to evaluate contracts:
	Given an ambiguous principal-agent setting, a contract's \emph{worst-case expected payoff} is its infimum expected payoff to the principal over all distributions $\{F_i\}_{i=1}^n$ compatible with the known expected outcomes $\{R_i\}_{i=1}^n$.
We follow \cite{Carroll15} in making the following assumption, which simplifies the results in Section \ref{sec:robustness} but does not qualitatively affect them:%
\footnote{As explained in \cite[][Footnote 2]{Carroll15}, Assumption A4 is simply an additive normalization of the principal's payoffs. Without this assumption, a robustly optimal contract would take the form $t_j = \alpha(x_j-x_1)$. Further justification for assuming $x_1=0$ is that the principal may have ambiguity not just with respect to the action distributions but also as to her possible rewards, and she prefers a contract robust to the possibility (however slim) of receiving a zero reward.} 
\begin{enumerate}[topsep=1ex,itemsep=0ex,parsep=0ex]
	\item[\bf A4] In ambiguous principal-agent settings, the outcome $0$ belongs to $\Omega_m$, i.e., $x_1=0$.
\end{enumerate} 

In Section \ref{sec:linear-contracts}, we are interested in bounding the potential loss in the principal's expected payoff if she is restricted to use a linear contract. Formally, let $\mathcal{A}$ be the family of principal-agent settings.
For $(A_n,\Omega_m) \in \mathcal{A}$, denote by $OPT(A_n,\Omega_m)$ the optimal expected payoff to the principal with an arbitrary contract, and by $ALG(A_n,\Omega_m)$ the best possible expected payoff with a contract of the restricted form (we omit $(A_n,\Omega_m)$ from the notation where clear from context). We seek to bound 
$
\rho(\mathcal{A}) := \max_{(A_n,\Omega_m) \in \mathcal{A}} \frac{OPT(A_n,\Omega_m)}{ALG(A_n,\Omega_m)}.
$

\section{Properties and Geometry of Linearly-Implementable Actions}
\label{sec:geometry}

We begin by developing a geometric characterization of \emph{linearly-implementable} actions, and deriving several useful consequences, which we apply in Section \ref{sec:robustness} and especially in Section \ref{sec:linear-contracts}. 
See Appendix~\ref{appx:geometry} for details and missing proofs.

	%
	%
	%
	%

Consider a principal-agent setting $(A_n,\Omega_m)$. We say that an action $a\in A_n$ is \emph{linearly-implementable} 
	if there exists a linear contract with parameter $\alpha\le 1$ that implements $a$. %
	Note that $\alpha \leq 1$ is w.l.o.g.~because the expected payoff to the principal would be negative otherwise.

\begin{observation}
	\label{obs:lin-utility-payoff}
	If a linear contract with parameter $\alpha$ implements action $a^*$, then the agent's expected 
	utility and the principal's expected payoff are, respectively,
	\begin{equation*}
	\alpha R_{a^*} - c_{a^*};~~~(1-\alpha) R_{a^*}.
	\end{equation*}
\end{observation}

\begin{observation}
	\label{obs:single-implemented}
	Under Assumptions A1-A3, a linear contract implements at most one action.
\end{observation}

\begin{observation}
	\label{obs:indifferent}
	Let $a,a'$ be a pair of actions such that $R_{a'}>R_a$ and $c_{a'}>c_a$. Then a linear contract with parameter $\alpha_{a,a'}=\frac{c_{a'}-c_{a}}{R_{a'}-R_{a}}$ makes the agent indifferent among actions $a'$ and $a$ (but does not necessarily incentivize either of these actions).
\end{observation}

Let $N$ denote the number of linearly-implementable actions, and let $I_N\subseteq A_n$ denote the set of such actions. 
Index the actions in $I_N$ in order of their expected outcomes, i.e., for $a_i,a_{i'}\in I_N$, $i<i'\iff R_i < R_{i'}$.
We now define two different mappings, and in Lemma \ref{lem:upper-envelope} establish their equivalence.

(1) \emph{Linear-implementability mapping $a(\cdot)$.} Denote by $a(\cdot):[0,1]\to I_N\cup\{\varnothing\}$ the mapping of $\alpha$ to either the action implemented by the linear contract with parameter~$\alpha$ (observe there is at most one such action under our assumptions---for completeness we state and prove this in Appendix \ref{appx:geometry}),
or to $\varnothing$ if there is no such action. So mapping $a(\cdot)$ is onto~$I_N$ by definition. 
Denote by $\alpha_i$ the smallest $\alpha\in[0,1]$ such that action $a_i\in I_N$ is implemented by a linear contract with parameter $\alpha$, then $\alpha_i$ is the smallest $\alpha$ such that $a(\alpha)=a_i$.

(2) \emph{Upper envelope mapping $u(\cdot)$.} For every action $a\in A_n$, consider the line $\alpha R_a - c_a$ and let $\ell_a$ denote the segment between $\alpha=0$ and $\alpha=1$; these segments appear in Figure~\ref{fig:upper-envelope}, where the $x$-axis represents the possible values of $\alpha$ from 0 to~1.
Take the upper envelope of the $n$ segments $\{\ell_a\}_{a\in A_n}$ and consider its non-negative portion.
Let $u(\cdot):[0,1]\to A_n\cup\{\varnothing\}$ be the mapping from $\alpha$ to either $\varnothing$ if the upper envelope is negative at $\alpha$, or to the action whose segment forms the upper envelope at $\alpha$ otherwise. If there is more than one such action, let $u(\alpha)$ map to the one with the highest expected outcome $R_a$. 

We begin by showing the following monotonicity property of the upper envelope mapping $u$. Let $\underline{\alpha}$ denote the smallest $\alpha$ at which the upper envelope intersects the $x$-axis (or $\underline{\alpha}=1$ if no such $\alpha$ exists). As $\alpha$ goes from $\underline{\alpha}$ to 1, $u(\alpha)$ maps to actions with increasingly higher expected outcomes $\{R_a\}$, costs~$\{c_a\}$ and expected welfares $\{R_a-c_a\}$. 

\begin{lemma}
	\label{lem:envelope-monotonicity}
	For every two parameters $0\le \alpha < \alpha' \le 1$, either $u(\alpha)=\varnothing$, or it holds that $u(\alpha)=a\in A_n$ and $u(\alpha')=a'\in A_n$ such that (i) $R_a<R_{a'}$; (ii) $c_a<c_{a'}$; and (iii) $R_a-c_a\le R_{a'}-c_{a'}$.
\end{lemma}

Our main structural insight in this section is that the upper envelope mapping precisely captures linear implementability:

\begin{lemma}%
	\label{lem:upper-envelope}
	For every $\alpha\in[0,1]$, $a(\alpha)=u(\alpha)$.
\end{lemma}

\begin{figure}[t]
	\centering
	\begin{tikzpicture}[xscale=0.7,yscale=0.55]
	\draw[draw=none, use as bounding box](-1,-2.5) rectangle (9,3.75);
	\clip (-1,-2.25) rectangle (11,5);
	\draw[->,thick] (-0.5,0) -- (8.5, 0) node[right] {$\alpha$};
	\draw[->,thick] (0,-4.5) -- (0, 4.5);
	\draw[->,thick] (8,-4.5) -- (8, 4.5) node[right] {~~~~$\alpha R_a - c_a$};
	\draw[-,thick] (0,-0.25) -- (8, 1);
	\draw[-,thick] (0,-1.5) -- (8, 2.5);
	\draw[-,thick] (0,-4) -- (8, 3.5);
	\draw[-,thick] (0.1,-0.25) -- (-0.1,-0.25) node[left] {$-c_1$};
	\draw[-,thick] (0.1,-1.5) -- (-0.1,-1.5) node[left] {$-c_2$};
	\draw[-,thick] (0.1,-4) -- (-0.1,-4) node[left] {$-c_3$};
	\draw[-,thick] (8.1,1) -- (7.9,1) node[right] {~~~~$R_1-c_1$};
	\draw[-,thick] (8.1,2.5) -- (7.9,2.5) node[right] {~~~~$R_2-c_2$};
	\draw[-,thick] (8.1,3.5) -- (7.9,3.5) node[right] {~~~~$R_3-c_3$};
	\draw[-,thick] (1/5*8,0.1) -- (1/5*8,-0.1) node[below] {$\alpha_1$};
	\draw[-,thick] (5/11*8,0.1) -- (5/11*8,-0.1) node[below] {$\alpha_2$};
	\draw[-,thick] (5/7*8,0.1) -- (5/7*8,-0.1) node[below] {$\alpha_3$};
	\draw[very thick,red] (1/5*8,0) -- (5/11*8,14/44);
	\draw[very thick,red] (5/11*8,14/44) -- (5/7*8,19/14);
	\draw[very thick,red] (5/7*8,19/14) -- (8,3.5);
	\end{tikzpicture}
	\caption{Linearly-implementable actions via upper envelope.}\label{fig:upper-envelope}
\end{figure}

Lemma \ref{lem:upper-envelope} has three useful implications:
(1)~The actions $\{a_i \in I_N\}$ appear on the upper envelope in the order in which they are indexed (i.e., sorted by increasing expected outcome). (2)~These actions are also sorted by increasing welfare, i.e., $R_1 - c_1 \leq R_2 - c_2 \leq \dots \leq R_N-c_N$. (3)~The smallest $\alpha$ that incentivizes action $a_i$ (which we refer to as $\alpha_i$) is the same $\alpha$ that makes the agent \emph{indifferent} between action $a_i$ and action $a_{i-1}$.
We denote this ``indifference $\alpha$'' by $\alpha_{i-1,i}$ and observe that $\alpha_{i-1,i}=\frac{c_{i}-c_{i-1}}{R_{i}-R_{i-1}}$ from Observation~\ref{obs:indifferent}. Using this notation, we can rewrite the third implication as $\alpha_i = \alpha_{i-1,i}$. Formally:   

\begin{corollary}
	\label{cor:what-a-returns}
	For every $\alpha\in[0,1]$, 
	\begin{eqnarray}
	\forall i\in [N-1]~~:~~a(\alpha)=a_i &\iff& \alpha_i\le \alpha <\alpha_{i+1},\label{eq:up-to-max}\\
	a(\alpha)=a_N &\iff& \alpha_N\le \alpha \le 1.\nonumber
	\end{eqnarray}
\end{corollary}

\begin{corollary}
	\label{cor:increasing-welfare}
	For every $i \in [N-1]$, $R_i - c_i \leq R_{i+1}-c_{i+1}$. 
\end{corollary}

\begin{corollary}
	\label{cor:alphas-equal-intersections}
	For every $i \in [N]$, $\alpha_i = \alpha_{i-1,i}$.
\end{corollary}

We apply Corollaries \ref{cor:what-a-returns}-\ref{cor:alphas-equal-intersections} to derive our approximation bounds in Section \ref{sec:linear-contracts}.

The final corollary in this section, which we apply to establish robustness in Section \ref{sec:robustness}, highlights the fact that in terms of linear-implementability, two principal-agent settings whose actions have the same expected outcomes and costs are in effect equivalent. The distributions, outcome values, and even number of outcomes matter for linear-implementability only to the extent of determining the expected outcome of each action. This property is special to linear contracts; in contrast, optimal contracts depend on the details of the distributions beyond their expected outcomes, and this adds to their complexity.

\begin{corollary}
	\label{cor:linear-equivalence}
	Consider two principal-agent settings $(A_n,\Omega_m)$, $(A'_n,\Omega'_{m'})$ for which there exists a bijection $b:A_n\to A'_n$ between the action sets, such that actions $a$ and $b(a)$ have the same expected outcome $R_{a}=R_{b(a)}$ and cost $c_{a}=c_{b(a)}$ for every $a\in A_n$. 
	Let $a,a'$ be the linearly-implementability mappings of the two settings, respectively.
	Then for every parameter $\alpha\in[0,1]$, $b(a(\alpha))=a'(\alpha)$, and the principal's expected payoff from a linear contract with parameter $\alpha$ is the same in both settings.
\end{corollary}	

\section{Robust Optimality of Linear Contracts}
\label{sec:robustness}

In this section we establish our robust optimality result for linear contracts. We show that a linear contract maximizes the principal's expected payoff in ambiguous settings, \emph{in the worst case over the unknown distributions}. 
All deferred proofs appear in Appendix \ref{appx:robustness}.

\begin{theorem}[Robust optimality]
	\label{thm:robustness}
	For every ambiguous principal-agent setting, an optimal linear contract has maximum worst-case expected payoff among all limited liability contracts. 
\end{theorem}

In Theorem \ref{thm:robustness}, ``an optimal linear contract'' is well-defined: For a (non-ambiguous) principal-agent setting, a linear contract is optimal when it maximizes the principal's expected payoff over all linear contracts. For an \emph{ambiguous} principal-agent setting, a linear contract has the same expected payoff to the principal over all compatible distributions---see Corollary~\ref{cor:linear-equivalence}. Thus an optimal linear contract can still be defined as maximizing the principal's expected payoff.
In the remainder of the section we prove Theorem~\ref{thm:robustness}.

\subsection{Main Lemma for Robust Optimality}
\label{sub:main-robustness-lemma}

\pdnote{The key step in the proof of Theorem~\ref{thm:robustness} is to show that we may restrict the search for optimally robust contracts to affine contracts.}

\begin{lemma}
	\label{lem:main-robustness-lemma}
	Consider an ambiguous principal-agent setting $(A_n,\Omega_m)$. For every limited liability contract with payment scheme $t$, there exist compatible distributions $\{F_i\}_{i=1}^n$ and an affine contract with parameters $\alpha_0, \alpha_1\ge 0$, such that the affine contract's expected payoff is at least that of contract~$t$ for distributions~$\{F_i\}_{i=1}^n$. 
\end{lemma}

	\begin{figure}[t]
		\centering
	\begin{tikzpicture}[scale=0.8]
	\draw[draw=none, use as bounding box](-2,-0.5) rectangle (6,4);
	\draw[->,thick] (0,0) -- (0,4) node[left] {$t$};
	\draw[->,thick] (0,0) -- (5,0) node[below] {$x$};
	\draw[-,dotted,thick] (0,2.5) -- (1.5,2.5);
	\draw[-,dotted,thick] (1.5,0) -- (1.5,2.5);
	\draw[-,dotted,thick] (0,3) -- (4,3);
	\draw[-,dotted,thick] (4,0) -- (4,3);
	\draw[-,thick] (0,1.5) -- node[above,yshift=-1.5mm,xshift=-2mm] {$l_2$} (1.5,2.5);
	\draw[-,thick] (0,1.5) -- node[below] {$l_1$} (4,3);
	\draw[-,thick] (1.5,2.5) -- node[above,xshift=2.5mm,yshift=1.5mm] {$l_3$} (4,3);
	\node[left] at (0,1.5) {$t_1$};
	\node[left] at (0,2.5) {$t_j$};
	\node[left] at (0,3) {$t_m$};
	\node[below] at (0,0) {$x_1$};
	\node[below] at (1.5,0) {$x_j$};
	\node[below] at (4,0) {$x_m$};	
	\end{tikzpicture}
	\begin{tikzpicture}[scale=0.8]
	\draw[draw=none, use as bounding box](-2,-0.5) rectangle (6,4);
	\draw[->,thick] (0,0) -- (0,4) node[left] {$t$};
	\draw[->,thick] (0,0) -- (5,0) node[below] {$x$};
	\draw[-,dotted,thick] (0.75,0) -- (0.75,1.78125);
	\draw[-,dotted,thick] (0,1.78125) -- (0.75,1.78125);
	\draw[-,dotted,thick] (2,0) -- (2,2.6);
	\draw[-,dotted,thick] (0,2.6) -- (2,2.6);
	\draw[-,dotted,thick] (-1.5,2.25) -- (2,2.25);
	\draw[-,thick] (0,1.5) -- (1.5,2.5);
	\draw[-,thick] (0,1.5) -- (4,3);
	\draw[-,thick] (1.5,2.5) -- (4,3);
	\node[left] at (0,1.78125) {$l_1(R_k)$};
	\node[left] at (-1.5,2.25) {$l_1(R_{i^*})$};
	\node[left] at (0,2.6) {$l_3(R_{i^*})$};
	\node[below] at (0.75,0) {$R_k$};
	\node[below] at (2,0) {$R_{i^*}$};
	\end{tikzpicture}
		\caption{{\bf Visualization of the proof of Lemma \ref{lem:main-robustness-lemma}---Case I.} In the case depicted here, point $(x_j,t_j)$ is strictly above line $l_1$. This case in addressed in Claim \ref{cla:above}.}\label{fig:robust-linear-case1}
	\end{figure}

\vspace*{-4pt}
\begin{proof}
	The payment scheme $t$ maps the outcomes $0=x_1<\dots< x_m$ to payments $t_1,\dots,t_m\ge 0$. 
	Consider the two extreme outcomes $x_1, x_m$ %
	and their corresponding payments $t_1,t_m$. We begin by defining simple compatible distributions $\{F'_i\}_{i=1}^n$ whose support is the extreme outcomes, as follows. For every distribution-ambiguous action $a_i$, set $F'_{i,m}:= \frac{R_i}{x_m}$
	(this is a valid probability since $R_i\le x_m$; otherwise there could not have been compatible distributions). Set $F'_{i,1}:=1-F'_{i,m}$ and let the other probabilities be zero. The expected outcome of distribution $F'_i$ is $F'_{i,1}x_1+F'_{i,m}x_m=R_i$. The defined distributions already enable us to prove the lemma for the case of $t_1>t_m$:
	
	\begin{claim}
		\label{cla:downward-slope}
		Lemma \ref{lem:main-robustness-lemma} holds for the case of $t_1>t_m$.
	\end{claim}

	\vspace*{-6pt}
	 \begin{myproof}[Proof of Claim \ref{cla:downward-slope}]
	 A proof of this claim appears in 
	 Appendix~\ref{appx:robustness}.
	 %
	 \end{myproof}
	 
	Assume from now on that $t_1\le t_m$. If $t$ is affine, this means that its slope parameter $\alpha_1$ must be non-negative. Similarly, we can write $t_1=\alpha_0+\alpha_1 x_1$ and plug in our assumption that $x_1=0$ to get $t_1=\alpha_0$, and so by limited liability ($t\ge 0$), $\alpha_0$ must also be non-negative. Thus if $t$ is affine, Lemma \ref{lem:main-robustness-lemma} holds. We focus from now on on the case that $t$ is non-affine; this guarantees the existence of a point $(x_j,t_j)$ as appears in Figures \ref{fig:robust-linear-case1} and \ref{fig:robust-linear-case2}. We argue this formally and then proceed by case analysis.
	
	\begin{claim}
		\label{cla:non-affine}
		If $t$ is non-affine, there exists an index $1<j<m$ such that the $3$ points $(x_1,t_1)$, $(x_j,t_j)$ and $(x_m,t_m)$ on the Euclidean plane are non-collinear.
	\end{claim} 

	 \vspace*{-6pt}
	 \begin{myproof}[Proof of Claim \ref{cla:non-affine}]
	 A proof of this claim appears in Appendix~\ref{appx:robustness}.
	 \end{myproof}

	We introduce the following notation -- denote the line between $(x_1,t_1)$ and $(x_m,t_m)$ by $l_1$, the line between $(x_1,t_1)$ and $(x_j,t_j)$ by $l_2$, and the line between $(x_j,t_j)$ and $(x_m,t_m)$ by $l_3$ (see Figures \ref{fig:robust-linear-case1} and \ref{fig:robust-linear-case2}).
	We denote the parameters of line $l_1$ by $\alpha_0$ and $\alpha_1$ (i.e., $t_1=\alpha_0 +\alpha_1 x_1$ and $t_m=\alpha_0 +\alpha_1 x_m$).
	These naturally give rise to a corresponding affine contract. 
	As argued above, since $t_1\le t_m$ we have $\alpha_1\ge 0$, and since $x_1=0$ and $t_1\ge 0$ we have $\alpha_0\ge 0$, so the affine contract has non-negative parameters. 
	
	Recall that the support of compatible distributions $\{F'_i\}_{i=1}^n$ is the endpoints of $l_1$. 
	We define alternative compatible distributions $\{F''_i\}_{i=1}^n$, whose support is either the endpoints of $l_2$ or of $l_3$, as follows: For every distribution-ambiguous action $a_i$, if $R_i\le x_j$ set $F''_{i,j}:=\frac{R_i}{x_j}$ (by assumption this is a valid probability), and $F''_{i,1}:=1-F''_{i,j}$. If $R_i>x_j$ set $F''_{i,m}:=\frac{R_i-x_j}{x_m-x_j}$ and $F''_{i,j}:=1-F''_{i,m}$. All other probabilities are set to zero. Observe that in either case, the expected outcome of distribution $F''_i$ is $R_i$. 
	With the two sets of compatible distributions $\{F'_i\}_{i=1}^n$ and $\{F''_i\}_{i=1}^n$ at hand, the analysis proceeds by addressing separately the two cases depicted in Figures \ref{fig:robust-linear-case1} and \ref{fig:robust-linear-case2}.

	\begin{figure}[t]
		\centering
		\begin{tikzpicture}[scale=0.8]
	\draw[draw=none, use as bounding box](-2,-0.5) rectangle (6,4);
	\draw[->,thick] (0,0) -- (0,4) node[left] {$t$};
	\draw[->,thick] (0,0) -- (5,0) node[below] {$x$};
	\draw[-,dotted,thick] (0,1) -- (1.5,1);
	\draw[-,dotted,thick] (1.5,0) -- (1.5,1);
	\draw[-,dotted,thick] (0,3) -- (4,3);
	\draw[-,dotted,thick] (4,0) -- (4,3);
	\draw[-,thick] (0,1.5) -- node[above,xshift=1.6mm,yshift=-1mm] {$l_2$} (1.5,1);
	\draw[-,thick] (0,1.5) -- node[above] {$l_1$} (4,3);
	\draw[-,thick] (1.5,1) -- node[below] {$l_3$} (4,3);
	\node[left] at (0,1.5) {$t_1$};
	\node[left] at (0,1) {$t_j$};
	\node[left] at (0,3) {$t_m$};
	\node[below] at (0,0) {$x_1$};
	\node[below] at (1.5,0) {$x_j$};
	\node[below] at (4,0) {$x_m$};	
	\end{tikzpicture}
		\caption{
			{\bf Visualization of the proof of Lemma \ref{lem:main-robustness-lemma}---Case II.} In the case depicted here,
			point $(x_j,t_j)$ is strictly below line $l_1$. This case in addressed in Claim \ref{cla:below}.}\label{fig:robust-linear-case2}
	\end{figure}

	The following simple observation will be useful in the case analysis:
	
	\begin{observation}%
	\label{obs:two-support}
		Consider a contract $t$ and two outcomes $x,x'$; let $l:\mathbb{R}\to\mathbb{R}$ be the line determined by points $(x,t(x))$ and $(x',t(x'))$.
		If action $a_i$ induces an outcome distribution $F_i$ over the support $\{x,x'\}$ with expectation $R_i$, then its expected payment to the agent is $l(R_i)$.
	\end{observation}

	\begin{claim}
		\label{cla:above}
		Lemma \ref{lem:main-robustness-lemma} holds for the case that $(x_j,t_j)$ is strictly above $l_1$.
	\end{claim}

	\vspace*{-6pt}
	\begin{myproof}[Proof of Claim \ref{cla:above}]		
	We prove the claim by showing that the affine contract with parameters $\alpha_0,\alpha_1$ (corresponding to $l_1$) has at least as high expected payoff to the principal as that of contract $t$ for distributions $\{F_i\}_{i=1}^n$, which are defined as follows:
	Let $a_{i^*}$ be the action incentivized by the affine contract with parameters $\alpha_0,\alpha_1$. Set $F_{i^*}=F''_{i^*}$ and $F_{k}=F'_{k}$ for every $k\ne i^*$.	
	We argue that for distributions $\{F_i\}_{i=1}^n$ as defined, contract $t$ will also incentivize action $a_{i^*}$, but at a (weakly) higher expected payment to the agent compared to the affine contract. 
	
	Consider first the affine contract. By definition, for every action $a_i$ with expected outcome $R_i$, the expected payment to the agent for choosing $a_i$ is $\alpha_0+\alpha_1 R_i$.
	Now consider contract $t$. For every action $a_k$ where $k\ne i^*$, by Observation \ref{obs:two-support} the expected payment to the agent taken over $F'_k$ is $l_1(R_k)=\alpha_0+\alpha_1 R_k$, i.e., the same as in the affine contract. For action $a_{i^*}$, by Observation~\ref{obs:two-support} the expected payment to the agent taken over $F''_k$ is either $l_2(R_{i^*})$ or $l_3(R_{i^*})$. In either case, since $l_2$ and $l_3$ are above $l_1$, the expected payment (weakly) exceeds $l_1(R_{i^*})$. Since $a_{i^*}$ maximizes the agent's expected utility in the affine contract, and the expected payment for $a_{i^*}$ only increases in contract $t$ while staying the same for other actions, this completes the proof of Claim \ref{cla:above}.
	\end{myproof}

	\begin{claim}
		\label{cla:below}
		Lemma \ref{lem:main-robustness-lemma} holds for the case that $(x_j,t_j)$ is strictly below $l_1$.
	\end{claim}

	 \vspace*{-6pt}
	 \begin{myproof}[Proof of Claim \ref{cla:below}]
	 	A proof of this claim appears in Appendix~\ref{appx:robustness}.
	 \end{myproof}

	This completes the analysis of the cases depicted in Figures \ref{fig:robust-linear-case1} and \ref{fig:robust-linear-case2}, thus proving Lemma~\ref{lem:main-robustness-lemma}. 
\end{proof}

\subsection{Proof of Theorem \ref{thm:robustness}}
\label{sub:robustness-wrap-up} 

\pdnote{Having established Lemma~\ref{lem:main-robustness-lemma}, Theorem~\ref{thm:robustness} is now easy to prove using the following observation:}

\begin{observation}
	\label{obs:affine-vs-linear}
	Consider an affine contract with parameters $\alpha_0,\alpha_1\ge 0$. For any distributions $\{F_i\}_{i=1}^n$, the expected payoff to the principal from the affine contract is at most the expected payoff from the linear contract with parameter $\alpha=\alpha_1$.
\end{observation}

\begin{corollary}%
	\label{cor:affine-vs-linear}
	For every affine contract with parameters $\alpha_0,\alpha_1\ge0$, there is a linear contract with (weakly) higher worst-case expected payoff. %
\end{corollary}

\begin{proof}[Proof of Theorem \ref{thm:robustness}]
Consider an ambiguous principal-agent setting. For every limited liability contract $t$, by Lemma \ref{lem:main-robustness-lemma} there exist compatible distributions $\{F_i\}_{i=1}^n$ and an affine contract with parameters $\alpha_0,\alpha_1\ge 0$ such that the worst-case expected payoff of contract $t$ is at most the expected payoff of the affine contract for distributions $\{F_i\}_{i=1}^n$. But the expected payoff of an affine contract is identical for all compatible distributions, and so for every limited liability contract $t$ there exists an affine contract with parameters $\alpha_0,\alpha_1\ge 0$ and higher worst-case expected payoff. By Corollary \ref{cor:affine-vs-linear}, the optimal linear contract has even higher worst-case expected payoff, completing the proof.  
\end{proof}

\section{Approximation Guarantees of Linear Contracts}
\label{sec:linear-contracts}

In this section we study linear contracts and their approximation guarantees. We show tight bounds on the approximation guarantees of linear contracts in all relevant parameters of the model---number of actions and outcomes (Section~\ref{sub:approx-num-actions}), spread of rewards (Section~\ref{sub:approx-range-E-rewards}) and costs (Section~\ref{sub:approx-range-costs}).
Our upper bounds in fact apply to the stronger benchmark of optimal welfare. Our lower bounds continue to hold if the instances are required to satisfy MLRP.
See Appendix~\ref{appx:approximation} for 
for all proofs omitted from this section.

\subsection{Bounds in the Number of Actions}
\label{sub:approx-num-actions}
\pdnote{Our first pair of results provides tight bounds on the approximation guarantees of linear contracts, as parameterized by the number of actions $n$.}

\begin{theorem}%
	\label{thm:linear-upper-bound}
	Consider a principal-agent setting $(A_n,\Omega_m)$ with $n$ actions and $N\le n$ linearly implementable actions. Then the multiplicative loss in the principal's expected payoff from using a linear contract rather than an arbitrary one is at most $N$. %
\end{theorem}

\begin{theorem}%
	\label{thm:lower-bound}
	For every $n$ and $\varepsilon>0$, 
	there is a principal-agent setting $(A_n,\Omega_m)$ with $n$ actions and $m = n$ outcomes for which MLRP holds,
	such that the multiplicative loss in the principal's expected payoff from using a linear contract rather than an arbitrary one is at least $n-\varepsilon$. %
\end{theorem}

The optimal contract in the construction used to prove Theorem~\ref{thm:lower-bound} is monotone. It follows that the gap between the best linear contract and the optimal monotone contract can $n$. In Appendix \ref{appx:LB-strong} we strengthen Theorem~\ref{thm:lower-bound} by showing that it applies even if $m = 3$, thus implying that the approximation ratio as parametrized by the number of outcomes $m$ can be unbounded.

\paragraph{Proof of upper bound.}

Our proof of Theorem~\ref{thm:linear-upper-bound} exploits the geometric insights developed in Section~\ref{sec:geometry}. The key tools are the following observation (Observation~\ref{obs:delta-welfare}) and the two lemmas (Lemma~\ref{lem:linear-bound-on-opt} and Lemma~\ref{lem:opt-less-than-welfare}) below.

Recall that $I_N$ denotes the set of $N\le n$ linearly-implementable actions, 
indexed such that their expected outcomes are increasing, i.e., $I_N=\{a_1,\dots,a_N\}$ 
and $R_1<\dots<R_N$. Note that by Assumption A1, $c_1<\dots<c_N$, and recall that this does \emph{not}
imply that $R_1 - c_1 \leq R_2 - c_2 \leq \dots \leq R_N-c_N$.

\begin{observation}
	\label{obs:delta-welfare}
	Consider two actions $a,a'$ such that $a$ has higher expected outcome and (weakly) higher welfare, i.e., $R_a > R_{a'}$ and $R_a-c_{a} \ge R_{a'}-c_{a'}$. Let $\alpha_{a',a}=\frac{c_a-c_{a'}}{R_a-R_{a'}}$. Then
	\begin{equation}
	(R_a - c_a) - (R_{a'} - c_{a'}) \leq (1-\alpha_{a',a}) R_a.\label{eq:intuition}
	\end{equation}
\end{observation}

\begin{proof}
	Since $R_a-c_{a} \ge R_{a'}-c_{a'}$ \pdnote{we have} $R_a-R_{a'} \ge c_a-c_{a'}$. Using that $R_a-R_{a'}>0$ we get $\alpha_{a',a}=\frac{c_a-c_{a'}}{R_a-R_{a'}}\le 1$. 
	So we can write $R_{a'}-c_{a'} \ge \alpha_{a',a}R_{a'}-c_{a'} = \alpha_{a',a}R_{a}-c_{a}$, where the equality follows from our definition of $\alpha_{a',a}$. 
	Hence, 
	$(R_a - c_a) - (R_{a'} - c_{a'}) \leq (R_a - c_a) - (\alpha_{a',a}R_a - c_a) = 
	(1-\alpha_{a',a}) R_a$, as required.
\end{proof}

Below we shall apply Observation \ref{obs:delta-welfare} to actions $a_i,a_{i-1}\in I_N$. In this context,
the intuition behind the observation is as follows: 
Consider the linear contract with parameter $\alpha_{i-1,i}\in[0,1]$. 
By Observation \ref{obs:indifferent}, in this contract 
the agent is indifferent among actions $a_i$ and $a_{i-1}$. The left-hand side of \eqref{eq:intuition} is the increase in expected welfare by switching to action $a_i$ from $a_{i-1}$. 
For the agent to get the same expected utility from $a_i$ and $a_{i-1}$, the principal must get this entire welfare increase as part of her expected payoff. The right-hand side of \eqref{eq:intuition} is the principal's expected payoff, and so the inequality holds.

\begin{lemma}
	\label{lem:linear-bound-on-opt}
	For every $k\in [N]$ and linearly-implementable action $a_k\in I_N$,
	\pdnote{$R_k - c_k \leq \sum_{i = 1}^{k} (1-\alpha_{i-1,i}) R_i.$}
\end{lemma}

\begin{proof}%
	The proof is by induction on $k$. 
	For $k=1$, recall that $\alpha_{0,1}=0$ by definition, and it trivially holds that $R_1-c_1\le R_1$.
	Now assume the inequality holds for $k-1$, i.e., 
	\begin{equation}
	R_{k-1} - c_{k-1} \leq \sum_{i = 1}^{k-1} (1-\alpha_{i-1,i}) R_i.\label{eq:induct-hypoth}
	\end{equation} 
	By Corollary \ref{cor:increasing-welfare}, the welfare of $a_k$ is at least that of $a_{k-1}$, and we know that $R_k>R_{k-1}$. We can thus apply Observation \ref{obs:delta-welfare} to actions $a=a_k,a'=a_{k-1}$ and get
	\begin{equation}
	(R_k - c_k) - (R_{k-1} - c_{k-1}) \leq (1-\alpha_{k-1,k}) R_k.\label{eq:induct-step}
	\end{equation}
	Adding inequality \eqref{eq:induct-step} to \eqref{eq:induct-hypoth} completes the proof for $k$.
\end{proof}

\begin{lemma}
	\label{lem:opt-less-than-welfare}
	Consider a principal-agent setting $(A_n,\Omega_m)$ with linearly-implementable action set $I_N\subseteq A_n$.
	The expected payoff of an optimal (not necessarily linear) contract is at most $R_N-c_N$.
\end{lemma}

\begin{proof}
	In a linear contract with parameter $\alpha=1$, the agent's expected utility for any action $a$ is its welfare $R_a-c_a$. 
	Thus an action is implemented by such a contract if and only if it maximizes welfare among all actions $A_n$. By Corollary~\ref{cor:what-a-returns}, $a(1)=a_N$ and so $a_N$ must be the welfare-maximizing action.
	In every contract, the IR property ensures that the agent's expected payment covers the cost $c_a$ of the implemented action $a$, and so the principal's expected payoff is always upper-bounded by $R_a-c_a$. We conclude that $OPT\le \max_{a\in A_n}\{R_a-c_a\}=R_N-c_N$, as required.
\end{proof}

\pdnote{With Lemma~\ref{lem:linear-bound-on-opt} and Lemma~\ref{lem:opt-less-than-welfare} at hand we can now prove Theorem~\ref{thm:linear-upper-bound}.} %

\begin{proof}[Proof of Theorem~\ref{thm:linear-upper-bound}]
\pdnote{To prove the approximation guarantee of $N \leq n$ claimed in the theorem, observe that}
	\begin{eqnarray}
	\pdnote{OPT \le R_N-c_N \label{eq:step1}
	\leq \sum_{i \le N} (1-\alpha_{i-1,i}) R_i
	= \sum_{i\le N} (1-\alpha_i) R_i
	\le N\cdot \max_{i\le N} \{(1-\alpha_i) R_i\}
	= N\cdot ALG,}
	\end{eqnarray}
	\pdnote{where the first inequality holds by Lemma \ref{lem:opt-less-than-welfare}, the second inequality holds by Lemma~\ref{lem:linear-bound-on-opt}, and the equality holds  by Corollary~\ref{cor:alphas-equal-intersections}.}
\end{proof}

\paragraph{Proof of lower bound.}

Our proof of Theorem~\ref{thm:lower-bound} relies on a recursively defined principal-agent setting with full information, in which the optimal contract can extract full welfare of $n$ while linear contracts can extract at most $1$. The construction requires the expected outcome and the costs to tend to infinity at different speeds.

\begin{proof}[Proof of Theorem \ref{thm:lower-bound}]
	For every $n$, consider a family of principal-agent instances $\{(A^\epsilon_n, \Omega_n) \mid \epsilon> 0\}$, each with $n$ actions and $m=n$ outcomes. For every $i\in[n]$, the $i$th action $a_i=(F_i,c_i)\in A^\epsilon_n$ has $F_{i,i}=1$, i.e., deterministically leads to the $i$th outcome $x_i\in\Omega_n$. 
	Every principal-agent instance is thus a full information setting in which the outcome indicates the action, and for which MLRP 
	holds.
	We define action $a_i$'s expectation $R_i$ (equal to outcome $x_i$) and its cost $c_i$ recursively:
	$$
	R_{i+1} = \frac{R_i}{\epsilon},~~~c_{i+1} = c_i + (R_{i+1}-R_i)\left(1-\frac{1}{R_{i+1}}\right),
	$$
	where $R_1 = 1$ and $c_1 = 0$. 
	
	We establish several useful facts about instance $(A^\epsilon_n, \Omega_n)$: For every $i\in[n]$, it is not hard to verify by induction that
	\begin{eqnarray}
	&R_{i} = \frac{1}{\epsilon^{i-1}},~~~
	c_{i} = \frac{1}{\epsilon^{i-1}}-i+\epsilon(i-1),&\label{eq:induction-vals1}\\
	&R_{i} - c_{i} = i-\epsilon(i-1).&\label{eq:induction-delta}
	\end{eqnarray}
	Observe that the actions are ordered such that $R_i$, $c_i$, and $R_i-c_i$ are strictly increasing in $i$. Plugging the values in \eqref{eq:induction-vals1} into $\alpha_{i-1,i}= \frac{c_{i}-c_{i-1}}{R_{i}-R_{i-1}}$ we get
	\begin{eqnarray}
	&\alpha_{i-1,i} = 1 - \epsilon^{i-1},\nonumber&\\
	&(1-\alpha_{i-1,i})R_{i} = 1.&\label{eq:induction}
	\end{eqnarray}
	$\alpha_{i-1,i}$ is also strictly increasing in $i$.
	
	Let $OPT^\epsilon$ (resp., $ALG^\epsilon$) denote the optimal expected payoff from an arbitrary (resp., linear) contract in the principal-agent setting $(A^\epsilon_n, \Omega_n)$.
	In a full information setting, the principal can extract the maximum expected welfare. This can be achieved by paying only for the outcome that indicates the welfare-maximizing action, and only enough to cover its cost. From \eqref{eq:induction-delta} we thus get
	$$
	OPT^\epsilon = \max_i\{i-\epsilon(i-1)\} = n-\epsilon(n-1) \underset{\epsilon\to 0}{\longrightarrow} n.
	$$ 
	In Lemma \ref{lem:aux-implementable} in Appendix \ref{appx:approximation}, we analyze
	linear-implementability in the setting $(A^\epsilon_n, \Omega_n)$,
	showing that $\alpha_i=\alpha_{i-1,i}$ for every $i\in[n]$. 
	Thus      
	from~\eqref{eq:induction} it follows that	
	$$
	ALG^\epsilon\le 1,
	$$ 
	completing the proof.
\end{proof}

\subsection{Bounds in the Range of Expected Rewards}
\label{sub:approx-range-E-rewards}
Our second pair of results is parametrized by the range of the expected outcomes $\{R_i\}$ normalized such that $R_i \in [1,H)$ for all $a_i \in A_n$. 
Consider bucketing these actions by their expected outcomes into $\lceil\log H\rceil$ buckets with ranges $[1,2),[2,4),[4,8)$, etc. Let $K$ be the number of non-empty buckets. 

\begin{theorem}%
	\label{thm:upper-bound-H}
	Consider a principal-agent setting $(A_n,\Omega_m)$ such that for every action $a\in A_n$, its expected outcome $R_a$ is $\in [1,H)$.
	The multiplicative loss in the principal's expected payoff from using a linear contract rather than an arbitrary one is at most $2K=O(\log H)$. %
\end{theorem}

\begin{corollary}%
	\label{cor:lower-bound-H}
	For every range $[1,H)$, there is a principal-agent setting $(A_n,\Omega_m)$ with $n=\log H$ actions and $m = n$ outcomes for which MLRP holds and $\forall a\in A_n:R_a\in [1,H)$, such that the multiplicative loss in the principal's expected payoff from using a linear contract rather than an arbitrary one is at least $\Omega(\log H)$. %
\end{corollary}

Our proof of Theorem~\ref{thm:upper-bound-H} applies an argument similar to the one used in Theorem~\ref{thm:linear-upper-bound} to the set of actions with the highest expected outcome in each bucket.

\begin{proof}[Proof of Theorem~\ref{thm:upper-bound-H}] 
	Recall the bucketing of actions in $I_N$ by their expected outcome.
	For every bucket $k\le K$, let $h(k)$ be the action with the highest $R_i$ in the bucket, and let $l(k)$ be the action with the lowest $R_i$. The bucketing is such that $R_{h(k)}/2< R_{l(k)}\le R_{h(k)}$. Since the actions in $I_N$ are ordered by their expected outcome, then $h(k)$ and $l(k)$ are increasing in $k$, and $h(k-1) + 1 = l(k)\le h(k)$. For the last bucket $K$ we have that $h(K)=a_N$, and
	by Lemma \ref{lem:opt-less-than-welfare} this implies $OPT\le R_{h(K)}-c_{h(K)}$. 

	Consider the subset of linearly-implementable actions $I_K= \{a_{h(k)}\mid k \in [K]\}\subseteq I_N$. These will play a similar role in our proof as actions~$I_N$ in the proof of Theorem \ref{thm:linear-upper-bound}. Let $\alpha_{h(k-1),h(k)}=(c_{h(k-1)}-c_{h(k)}) / (R_{h(k-1)}-R_{h(k)})$. %
	Observation \ref{obs:delta-welfare} applies to actions in $I_K$, and so we can apply a version of Lemma \ref{lem:linear-bound-on-opt} to get
	\begin{equation}
	OPT \le R_{h(K)}-c_{h(K)} \leq \sum_{k\le K} (1-\alpha_{h(k-1),h(k)}) R_{h(k)}.\label{eq:OPT-UB}
	\end{equation}
	Our goal is now to upper-bound the right-hand side of \eqref{eq:OPT-UB}. 
	
	\begin{claim}
		\label{cla:alpha-UB}
		$\alpha_{h(k-1),h(k)} \ge \alpha_{l(k)}.$
	\end{claim}

	\vspace*{-6pt}
	\begin{myproof}[Proof of Claim \ref{cla:alpha-UB}]
	Assume for contradiction that $\alpha_{h(k-1),h(k)} < \alpha_{l(k)}$.
	By definition of $\alpha_{h(k-1),h(k)}$ we have that $\alpha_{h(k-1),h(k)}R_{h(k)}-c_{h(k)} = \alpha_{h(k-1),h(k)}R_{h(k-1)}-c_{h(k-1)}$.
	Substituting $h(k-1)=l(k)-1$ we get
	\begin{equation}
	\alpha_{h(k-1),h(k)} R_{h(k)}-c_{h(k)} = \alpha_{h(k-1),h(k)} R_{l(k)-1}-c_{l(k)-1}.\label{eq:log-equation}
	\end{equation}
	
	Since the expected outcomes of actions in $I_K$ are strictly increasing, it holds that $R_{h(k)}>R_{l(k)-1}$, and so replacing $\alpha_{h(k-1),h(k)}$ with the larger $\alpha_{l(k)}$ in~\eqref{eq:log-equation} gives $\alpha_{l(k)}R_{h(k)}-c_{h(k)} > \alpha_{l(k)}R_{l(k)-1}-c_{l(k)-1}$.
	By Corollary~\ref{cor:alphas-equal-intersections}, $\alpha_{l(k)} = \alpha_{l(k)-1,l(k)}$, and so the right-hand side $\alpha_{l(k)}R_{l(k)-1}-c_{l(k)-1}$ equals $\alpha_{l(k)}R_{l(k)}-c_{l(k)}$. We conclude that $\alpha_{l(k)}R_{h(k)}-c_{h(k)} > \alpha_{l(k)}R_{l(k)}-c_{l(k)}$, i.e., in a linear contract with parameter $\alpha_{l(k)}$, action $h(k)$ has higher expected utility for the agent than action $l(k)$. But by definition, parameter $\alpha_{l(k)}$ implements action $l(k)$, and so we have reached a contradiction.
	\end{myproof}

	Applying Claim \ref{cla:alpha-UB} to \eqref{eq:OPT-UB} we get the following chain of inequalities: 
	\begin{align*}
	OPT \leq \sum_{k\le K} (1-\alpha_{l(k)}) R_{h(k)}
	< 2 \sum_{k\le K} (1-\alpha_{l(k)}) R_{l(k)} 
	\le  2K \cdot \max_{k \le K} \left\{(1-\alpha_{l(k)})R_{l(k)}\right\}
	\le 2K \cdot ALG,
	\end{align*}
	where the strict inequality follows from $R_{h(k)}/2< R_{l(k)}$.
\end{proof}

Corollary~\ref{cor:lower-bound-H} can be proven using a similar recursively defined instance as used in the proof of Theorem~\ref{thm:lower-bound}, except that instead of letting the expected outcomes and costs tend to infinity we double them.

\begin{proof}[Proof of Corollary \ref{cor:lower-bound-H}]
	We use the same construction as in the proof of Theorem \ref{thm:lower-bound}, but set $\epsilon=1/2$. Since $n=\log H$, it indeed holds that $R_n=2^{n-1}<H$. We know $OPT$ can achieve at least $n-\epsilon(n-1)> \frac{1}{2}n$ while $ALG$ can't do better than~$1$, completing the proof.
\end{proof}

\subsection{Bounds in the Range of Costs}
\label{sub:approx-range-costs}
Our final pair of results concerns the costs. As in the case of expected outcomes, suppose costs are normalized such that $c_i \in [1,C)$ for all $a_i \in A_n$, consider bucketing these into $\lceil C \rceil$ buckets $[1,2),[2,4), ...$ etc., and let $L$ be the number of non-empty buckets.

\begin{theorem}%
	\label{thm:upper-bound-C}
	Consider a principal-agent setting $(A_n,\Omega_m)$ such that for every action $a\in A_n$, its cost $c_a$ is $\in [1,C)$.
	The multiplicative loss in the principal's expected payoff from using a linear contract rather than an arbitrary one is at most $4L=O(\log C)$. %
\end{theorem}

\begin{corollary}%
	\label{cor:lower-bound-C}
	For every range $[1,C)$ such that $\log(2C) \geq 3$, there is a principal-agent setting $(A_n,\Omega_m)$ with $n=\log (2C)$ actions and $m = n$ outcomes for which MLRP holds and $\forall a\in A_n: c_a\in [1,C)$, such that the multiplicative loss in the principal's expected payoff from using a linear contract rather than an arbitrary one is at least $\frac{1}{2}n = \Omega(\log C)$.
\end{corollary}

\section{Beyond Linear Contracts}
\label{sec:beyond-linear}

We conclude by showing that a similar lower bound to our bound for linear contracts (Theorem \ref{thm:lower-bound}) applies to all \emph{monotone} contracts; the only difference is that our lower bound of $n$ (the number of actions) is slightly relaxed to $n-1$.
That is, we construct 
a MLRP instance
with $n$ actions in which the best monotone contract cannot guarantee better than a $\frac{1}{n-1}$-approximation to the optimal contract's expected payoff. 

\begin{theorem}
	\label{thm:mon-lower-bound}
	For every number of actions $n$, there is a principal-agent setting $(A^{\epsilon,\delta}_n,\Omega_m)$ parameterized by $\epsilon\gg \delta>0$ for which MLRP holds, such that the multiplicative loss in the principal's expected payoff from using a monotone contract rather than an arbitrary one approaches $n-1$ as $\epsilon,\delta\to 0$.	
\end{theorem} 

The class of monotone contracts captures in particular \emph{debt} contracts (mentioned already in Footnote~\ref{ftnt:debt}).  
Theorem \ref{thm:mon-lower-bound} thus implies that our results do not qualitatively change for this alternative family of simple contracts, which are common, e.g.,~in startup venture funding.%
\footnote{The first rewards generated by the startup CEO (the agent) go entirely to the principal to pay back for the initial loan that funded the company, so the contractual payments for the lower outcomes are zero; any further rewards are split linearly among the principal and agent.} On the upper bound side, it follows from Proposition \ref{pro:mlrp-increasing} (Appendix \ref{appx:simple-and-optimal}) that there always exists a debt contract which achieves a $\frac{1}{n-1}$-approximation to the optimal contract's expected payoff. 

\section{Conclusion}
\label{sec:conclusion}

One of the major contributions of \pdnote{theoretical computer science} 
to economics has been the use of approximation guarantees to systematically explore
complex economic design spaces, and to identify ``sweet spots'' of the
design space where there are plausibly realistic solutions that simultaneously
enjoy rigorous performance guarantees.  For example, in auction and
mechanism design, years of fruitful work by dozens of researchers has
clarified the power and limitations of ever-more-complex mechanisms in
a wide range of settings.  Contract theory presents another huge
opportunity for expanding the reach of the theoretical computer
science toolbox, and we believe that this paper takes a promising
first step in that direction.

\section*{Acknowledgments}
We wish to thank Michal Feldman and S.~Matthew Weinberg for enlightening conversations. %
This research was supported by the National Science Foundation (grants CCF-1524062 and CCF-1813188), the European Research Council (grant 740282), the ISRAEL SCIENCE FOUNDATION (grant No.~336/18), and the British Academy (grant SRG1819/191601). The third author is a Taub Fellow (supported by the Taub Family Foundation).

\bibliographystyle{abbrvnat}
\bibliography{abb,contracts-bib}

\appendix

\section{Linear Programming Formulation and Implications}
\label{appx:lp-results}

The LP for incentivizing action $a$ at minimum expected payment has $m$ payment
variables $\{t_j\}$, which by limited liability must be nonnegative,
and $n-1$ IC constraints ensuring that the agent's expected utility
from action $a$ is at least his expected utility from any other
action.  Note that by Assumption~A3, there is no need for an IR
constraint to ensure that the expected utility is nonnegative. The LP
is:
\begin{eqnarray}
\min & \sum_{j\in[m]} {F_{a,j}t_j} & \label{LP:min-pay}\\
\text{s.t.} & \sum_{j\in[m]} {F_{a,j}t_j} - c_a \ge \sum_{j\in[m]} {F_{a',j}t_j} - c_{a'} & \forall a'\ne a, a'\in A_n, \nonumber\\
& t_j \ge 0 & \forall j\in[m].\nonumber
\end{eqnarray}
The dual of LP \eqref{LP:min-pay} has $n-1$ nonnegative variables, one for every action other than $a$:
\begin{eqnarray}
\max & \sum_{a'\ne a} {\lambda_{a'}(c_a-c_{a'})} & \label{LP:dual}\\
\text{s.t.} & \sum_{a'\ne a} {\lambda_{a'}(F_{a,j}-F_{a',j})} \le F_{a,j} & \forall j\in[m], \nonumber\\
& \lambda_{a'} \ge 0 & \forall a'\ne a,a'\in A_n.\nonumber
\end{eqnarray}

In the following subsections, we first use the LP-based approach to provide additional details for Example \ref{ex:bad-example}. We then show two implications of the LP-based approach. First, the LP and its dual can be used to characterize if an action is implementable or not, whether by an arbitrary contract or by a monotone one. Second, there always exists an optimal contract with at most $n-1$ positive payments.

\subsection{Nonmonotonicity of the Optimal Contract}
\label{appx:example-analysis}

We analyze Example \ref{ex:bad-example} in Section \ref{sec:intro}, which demonstrates nonmonotonicity of the optimal contract.

For payment profile $t\approx(0,0,0.15,3.93,2.04,0)$ all IC constraints in the LP are tight, i.e., the agent's utility is the same for all actions. The agent tie-breaks in favor of action $a_3$, which has the highest expected payoff of $2.95 = 4.99 - 2.04$ for the principal (where $4.99$ is the expected outcome and $2.04$ the expected payment). 

No other action can achieve expected payoff for the principal as high as $a_3$. Moreover, if the payments are constrained to be monotone, or if the number of positive payments is constrained to be $<3$, then~$a_3$ can no longer be implemented for expected payment of merely $2.04$. 

\subsection{Implementable and Monotonically-Implementable Actions}
\label{sub:implement}

The following propositions characterize when actions are implementable, both by an arbitrary contract and by a monotone one. 

\begin{proposition}[Implementability]
	\label{pro:implementable}
	An action $a$ is implementable (up to tie-breaking) if and only if there is no convex combination of the other actions that results in the same distribution $\sum_{a'\ne a}{\lambda_{a'}F_{a'}}=F_a$ but lower cost $\sum_{a'\ne a}{\lambda_{a'}c_{a'}}<c_a$.
\end{proposition}

Proposition \ref{pro:implementable} is immediate from linear programming duality of LP \eqref{LP:min-pay} with the objective replaced by ``$\min 0$'', and indeed can be found in the economic literature.

\begin{proposition}[Monotonic implementability]
	\label{pro:mon-implementable}
	An action $a$ is implementable by a monotone contract (up to tie-breaking) if and only if there is no convex combination of the other actions that results in a first-order stochastically dominating distribution $\sum_{a'\ne a}{\lambda_{a'}F_{a'}}$ with lower cost $\sum_{a'\ne a}{\lambda_{a'}c_{a'}}<c_a$.
\end{proposition}

Proposition \ref{pro:mon-implementable} follows from linear programming duality of the following LP:

\begin{eqnarray*}
	\min & 0 & \\
	\text{s.t.} & \sum_{j\in[m]} {F_{a,j}t_j} - c_a \ge \sum_{j\in[m]} {F_{a',j}t_j} - c_{a'} & \forall a'\in A_n\setminus \{a\}, \nonumber\\
	& t_{j} \ge t_{j-1} & 2\le j \le m,\\
	& t_j \ge 0 & \forall j\in[m].\nonumber
\end{eqnarray*}

The dual has $n-1$ nonnegative $\lambda$-variables, one for every action other than $a$, and $m-1$ nonnegative $\mu$-variables $\mu_2,\dots,\mu_m$. It can be written w.l.o.g.~as:
\begin{eqnarray*}
	\max & c_a - \sum_{a'\ne a} {\lambda_{a'}c_{a'}} & \\
	\text{s.t.} & F_{a,1} \le \sum_{a'\ne a} {\lambda_{a'}F_{a',1}} + \mu_2 & , \nonumber\\
	& F_{a,j} \le \sum_{a'\ne a} {\lambda_{a'}F_{a',j}} + \mu_{j+1} - \mu_{j} & 2\le j \le m-1, \nonumber\\
	& F_{a,m} \le \sum_{a'\ne a} {\lambda_{a'}F_{a',j}} - \mu_m & , \nonumber\\
	& \sum_{a'\ne a} {\lambda_{a'}} = 1 & ,\\
	& \lambda_{a'} \ge 0 & \forall a'\ne a,a'\in A_n.\\
	& \mu_j \ge 0 & 2\le j \le m.
\end{eqnarray*}

\subsection{Number of Nonzero Payments in Optimal Contract}
\label{sub:nonzero-payments}

The next lemma implies the existence of an optimal contract with at most $n-1$ nonzero payments. 

\begin{lemma}
	\label{lem:pos-pay}
	Consider a principal-agent setting $(A_n,\Omega_m)$ with $n$ actions and $m$ outcomes.
	For every implementable action $a$, there is an implementing contract with minimum expected payment, such that its payment scheme is positive for $\le n-1$ outcomes.
\end{lemma}

\begin{proof}[Proof (for completeness)]
	Because we assume action $a$ is implementable, LP \eqref{LP:min-pay} is feasible and bounded, and so has an optimal basic feasible solution for which $m$ constraints are tight \cite{MatousekGartner07}. 
	There are only $n-1$ constraints other than non-negativity constraints, so at least $m-n+1$ of the non-negativity constraints are tight, meaning that the corresponding payments equal zero. Thus at most $n-1$ payments can be positive.
\end{proof}


\section{Regularity Assumptions in the Contract Theory Literature}\label{appx:regularity-assumptions}

\subsection{The Regularity Assumption of MLRP}
\label{appx:MLRP}

The economic literature (see, e.g., \cite{GrossmanHart83}) introduces
a regularity assumption called the \emph{monotone likelihood ratio
	property} (MLRP) for principal-agent settings. Intuitively, the
assumption asserts that the higher the outcome, the more likely it is
to be produced by a high-cost action than a low-cost one (in a
relative sense).  We
adapt the standard definition to accommodate for zero probabilities,
as follows:

\begin{definition}[MLR]
	\label{def:increasing-LR}
	Let $F,G$ be two distributions over $m$ values $v_1,\dots,v_m$. 
	The likelihood ratio $F_j/G_j$ 
	is monotonically increasing in $j$ if 
	$$
	F_j/G_j \le F_{j'}/G_{j'}
	$$
	for every $j<j'$ such that at least one of $F_j,G_j$ is
	positive, and at least one of $F_{j'},G_{j'}$ is positive.
\end{definition}

\begin{definition}[MLRP]
	\label{def:MLRP}
	A principal-agent problem satisfies {\em MLRP} if for every pair of actions $a,a'$ such that $c_a<c_{a'}$, the likelihood ratio $F_{a',j}/F_{a,j}$ is monotonically increasing in $j$. 
\end{definition}

\begin{proposition}[MLR $\implies$ FOSD \cite{TadelisSegal05}]
	If the likelihood ratio $F_j/G_j$ is monotonically increasing in $j$, then $F$ first-order stochastically dominates $G$. The converse does not hold.
\end{proposition}

We demonstrate MLRP and non-MLRP through the following examples:

\begin{example}[Two outcomes] 
	\label{ex:2-outcomes-MLRP}
	Assume there are $m=2$ outcomes $\ell<h$, and that a higher action cost means higher expected outcome. Then the probability $F_{i,h}$ of action $i$ to achieve the high outcome is strictly increasing in $i$. Thus MLRP holds.
\end{example}

\begin{example}[Spanning condition \cite{GrossmanHart83}]
	Assume that the agent has two basic actions, such as ``effort'' (costly) and ``no effort'' (cost zero), for which MLR holds, and the agent can interpolate among these (this is known as a setting satisfying the ``spanning condition''). Then the resulting action set satisfies MLRP.
\end{example}

\begin{example}[Binomial distributions]
	Assume that higher cost means more effort on behalf of the agent, that the level of effort determines the probability of the agent's success in a Bernoulli trial, and that the outcome is the Binomially distributed number of successful trials out of a total of $m-1$ trials. Then the action set satisfies MLRP: one can verify that 
	$$
	\frac{{{m-1}\choose{j-1}} p^{j-1} (1-p)^{m-j}} {{{m-1}\choose{j-1}} q^{j-1} (1-q)^{m-j}}
	$$
	is increasing in $j$ when $p>q$.
\end{example}

\begin{example}[No MLRP \cite{CaillaudHermalin00}]
	\label{ex:no-MLRP}
	Let $n=2$ and $m=3$. We define a principal-agent setting $(A_n,\Omega_m)$ where $A_n=(a_1,a_2)$ and $\Omega_m=(x_1,x_2,x_3)=(0, 1, 2)$. The actions are defined as follows:
	\begin{eqnarray*}
		a_1=(F_1,c_1) &=& ((1/3, 1/3, 1/3), ~0),\\ 
		a_2=(F_2,c_2) &=& ((1/3, 1/6, 1/2), ~1).
	\end{eqnarray*}
	These distributions can be viewed as convex combinations of distributions with MLR (analogously to combinations of regular distributions leading to irregularity in auction theory).
	Namely:
	$F_i = \frac{1}{3}F^1_{i} + \frac{2}{3}F^2_{i}$, where $F^1_{1}=(1,0,0)$, $F^2_{1}=(0,1/2,1/2)$, $F^1_{2}=(1,0,0)$, $F^2_{2}=(0,1/4,3/4)$, and the likelihood ratios $F^1_{2,j}/F^1_{1,j}$ and $F^2_{2,j}/F^2_{1,j}$ are both monotonically increasing.%
	\footnote{A possible economic story behind this example could be that the agent chooses between ``no effort'' (action $a_1$) and ``effort'' (action $a_2$). Without effort, the distribution is $(1/2,1/2)$ over outcomes $(x_2, x_3)$, and with effort the distribution is $(1/4,3/4)$. However, regardless of the agent's effort level, with probability $1/3$ some exogenous bad event occurs (e.g., the market adopts a different technology as the industry standard, causing sales to drop), resulting in an outcome of $x_1=0$.}
\end{example}

\subsection{The Regularity Assumption of CDFP}

In addition to MLRP, there are other less natural regularity assumptions in the literature, for example the following \emph{concavity of distribution function property} (CDFP).

\begin{definition}
	\label{def:CDFP}
	An action $a$ satisfies the \emph{concavity of distribution function property (CDFP)} if for every two actions such that $a$'s cost $c_a$ is a convex combination of their costs, it holds that $a$'s distribution over outcomes first-order stochastically dominates the convex combination of their distributions. A principal-agent setting satisfies CDFP if it holds for every action. 
\end{definition}


\section{Special Cases in Which Simple Contracts are Optimal}
\label{appx:simple-and-optimal}

In several special cases of interest, a simple contract with a single nonzero payment is optimal. 
In this appendix we highlight four such cases. 
An obvious one (given Lemma \ref{lem:pos-pay}, which bounds the number of non-negative payments by $n-1$) is that of $n=2$ actions. In Proposition \ref{pro:partition-2-acts} we state for completeness the optimal contract for $n=2$, and in Propositions~\ref{pro:cdfp}, \ref{pro:two-outcomes} and \ref{pro:mlrp-increasing} we identify three additional classes of principal-agent settings in which simple is optimal. 

\begin{proposition}
	\label{pro:partition-2-acts}
	For $n=2$ actions, if the optimal contract incentivizes the nonzero-cost action $a_2$ rather than $a_1$ with cost $0$, then there is an optimal contract that pays only for the outcome that maximizes the likelihood ratio $F_{2,j}/F_{1,j}$, and the payment is $\frac{c_2}{F_{2,j}-F_{1,j}}$.
\end{proposition}

\begin{proof}
	Consider an optimal solution with a single positive payment (Lemma \ref{lem:pos-pay}).
	The constraint $\sum_j F_{1,j}p_j-c=\sum_j F_{0,j}p_j$ must be tight. For every $j$, if $p_j>0$ and the rest of the payments are zero, then $p_j=\frac{c}{F_{1,j}-F_{0,j}}$ and the expected payment is $F_{1,j}p_j=\frac{c}{1-F_{0,j}/F_{1,j}}$. To minimize this expected payment, $F_{0,j}/F_{1,j}$ must be minimized, completing the proof.
\end{proof}

\begin{proposition}\label{pro:cdfp}
	\label{pro:mlrp-and-cdfp}
	Consider a principal-agent setting $(A_n,\Omega_m)$ with $n$ actions and $m$ outcomes. Then there exists an optimal contract with a single nonzero payment for the highest outcome $x_m$ if the setting satisfies MLRP and there exists an action that is implementable by an optimal contract satisfying CDFP. 
\end{proposition}

To prove this proposition we need the following lemma.

\begin{lemma}
	\label{lem:highest-action}
	Consider a principal-agent setting $(A_n,\Omega_m)$ with $n$ actions and $m$ outcomes, for which MLRP holds. If the highest-cost action $a_n$ is implementable, then 
	there is an implementing contract with minimum expected payment that has a single nonzero-payment, which is rewarded for the highest outcome $x_m$.
\end{lemma}

\begin{proof}
	Recall the implementability primal LP from Appendix \ref{appx:lp-results}. Its variables are the payments $t_1,\dots,t_m$. We need to show that there is an optimal solution to this LP for action $a_n$ 
	with a single nonzero variable $t_m$. 
	We achieve this by creating a \emph{reduced} version of the primal LP with only one variable $t_m$, and showing that its optimal objective value is no worse (no larger) than that of the original LP. Our argument uses the dual LP, in which there is a constraint for every one of the $m$ outcomes. Our proof proceeds as follows:  
	\begin{itemize}
		\item[1.] Create a reduced dual by dropping all the constraints except for the one corresponding to the maximum outcome $m$.
		\item[2.] Solve the relaxed dual LP to optimality.
		\item[3.] Verify that the resulting solution is feasible (and hence optimal) for the original dual LP.
	\end{itemize}
	
	These steps are sufficient to complete the proof: The reduced dual has the same optimal objective value as the original one. Dualizing back, the optimal value of the reduced primal LP is the same as that of the original primal LP, as required. Hence there is an optimal primal solution that only uses~$t_m$.
	
	\paragraph{Step 1:}
	The reduced dual is:
	\begin{eqnarray}
	\max & \sum_{i<n} {\lambda_{i}(c_n-c_{i})} & \nonumber\\
	\text{s.t.} & \sum_{i<n} {\lambda_{i}(F_{n,m}-F_{i,m})} \le F_{n,m} &, \nonumber\\
	& \lambda_{i} \ge 0 & \forall i< n.\nonumber
	\end{eqnarray}
	
	\paragraph{Step 2:}
	To solve the reduced dual, we note that $(c_n-c_i) \ge 0$ for every $i < n$ (using that $a_n$ is the highest-cost action). Thus, all coefficients in our objective are nonnegative.   Also, since MLRP implies stochastic dominance, $(F_{n,m} - F_{i,m}) \ge 0$ for every $i < n$.  Thus, all coefficients in our (sole) dual constraint are nonnegative.  The optimal solution is then to ``max out'' on the action with the maximum ``bang-per-buck,'' meaning an action in $\arg\max_i (c_n-c_i)/(F_{n,m}-F_{i,m})$.  To make the dual constraint tight, we set $\lambda_i = F_{n,m}/(F_{n,m}-F_{i,m})$ (and other variables to 0).
	
	\paragraph{Step 3:}
	In this step we need to verify feasibility of the original dual constraints.  Pick an arbitrary outcome $j < m$.  With our choice of $\lambda_i$, feasibility becomes $F_{n,m}(F_{n,j}-F_{i,j})/(F_{n,m}-F_{i,m}) \le F_{n,j}$.  After clearing denominators and canceling terms, this reduces to MLRP (i.e., $F_{n,m}/F_{i,m} \ge F_{n,j}/F_{i,j}$), which holds by assumption.
\end{proof}

We are now ready to prove the proposition.

\begin{proof}[Proof of Proposition~\ref{pro:cdfp}]
	The proof follows that of Proposition 12 in \cite{CaillaudHermalin00}.
	Assume there exists an optimal contract with payment profile $t$ implementing a CDFP action $a_i$ at expected cost $T$. Denote the principal's expected payoff by $OPT$.
	We show there is simple contract of the required format implementing $a_i$ at the same expected cost (thus achieving $OPT$).
	
	Consider a new setting with actions $\{a_{k}\mid k\le i\}$. 
	Action $a_i$ is implementable in this setting (e.g., by transfer profile $t$) and MLRP holds. 
	Thus by Lemma \ref{lem:highest-action}, there is a contract implementing $a_i$ at minimum expected cost with nonzero payment only for $x_m$. Denote this payment by $t'_m$, and let $T'$ be its expected cost. Since removing actions could not have increased the cost of implementing $a_i$, $T'\le T$. Denote the principal's expected payoff by $OPT'$ and observe $OPT'\ge OPT$.
	Let $i'<i$ be such that the IC constraint is binding for action $a_{i'}$ (such an action must exist or $T'$ could have been lowered).

	Now add back the actions $\{a_{k}\mid k>i\}$.  We argue that the same simple contract still implements~$a_i$. 
	Indeed, assume for contradiction that the simple contract implements $a_{i''}$ rather than $a_i$, where $i''>i$. 
	We will use CDFP to show that in this case, the principal's payoff is $>OPT'$, in contradiction to the optimality of the contract achieving $OPT\le OPT'$. 
	
	Since $c_{i'} < c_i < c_{i''}$, we can write $c_i=\lambda c_{i'} + (1-\lambda)c_{i''}$ for $\lambda\in[0,1]$. By CDFP, the distribution of~$a_i$ first-order stochastically dominates that of the mixed action $\lambda a_{i'} + (1-\lambda) a_{i''}$. Since the simple contract only pays for the highest outcome, this means that the agent's utility from $a_i$ is at least his utility from $\lambda a_{i'} + (1-\lambda) a_{i''}$. By tightness of the IC constraint for $a_{i'}$ and by the agent's preference for $a_{i''}$, the IC constraint must be tight for $a_{i''}$ as well. We know that the agent breaks ties in favor of the principal, and so the principal's payoff from $a_{i''}$ must exceed her payoff of $OPT'$ from $a_i$. 
\end{proof}

\begin{proposition}\label{pro:two-outcomes}
Consider a principal-agent setting $(A_n,\Omega_m)$ with $n$ actions and $m$ outcomes. Then there exists an optimal contract with a single nonzero payment for the highest outcome $x_m$ if there are $m = 2$ outcomes.
\end{proposition}

\begin{proof}
	Assume there are 2 outcomes, low ($\ell$) and high ($h$). There are $n$ actions numbered in (strictly) increasing order of their expected outcome $R_i$. 
	Observe that the probability $F_{i,h}$ of action $i$ to achieve the high outcome is also strictly increasing. This implies that MLRP holds. 
	Let $a_i$ be an optimally-implementable action. If $i=1$ (so $a_i$ is the zero-cost action), the trivial contract with no payments is optimal; if $i=n$, the claim follows from Lemma \ref{lem:highest-action}. 
	We now show that if $1<i<n$, action $a_i$ must satisfy CDFP, and so the proof is complete by Proposition~\ref{pro:cdfp}.
	
	Assume for contradiction that $a_i$ is not CDFP. So there exist two actions $a_{i'},a_{i''}$ where $i' < i < i''$, and $\lambda\in[0,1]$, such that 
	\begin{eqnarray*}
	c_i &=& \lambda c_{i'} + (1-\lambda)c_{i''},\\
	F_{i,h} & < & \lambda F_{i',h} + (1-\lambda)F_{i'',h}.
	\end{eqnarray*}
	This implies existence of $\lambda < \lambda' < 1$ such that
	\begin{eqnarray*}
	c_i &>& \lambda' c_{i'} + (1-\lambda')c_{i''},\\
	F_{i,h} &=& \lambda' F_{i',h} + (1-\lambda')F_{i'',h}.
	\end{eqnarray*}
	But by Proposition \ref{pro:implementable} this means that action $a_i$ is not implementable, contradiction.
\end{proof}

\begin{proposition}
	\label{pro:mlrp-increasing}
	Consider a principal-agent setting $(A_n,\Omega_m)$ with $n$ actions and $m$ outcomes. Then there exists an optimal contract with a single nonzero payment for the highest outcome $x_m$ if the setting satisfies MLRP and the actions have strictly increasing welfare.  
\end{proposition}

\begin{proof}
	Assume there is an optimal contract with payment profile $t$ incentivizing $a_i$. Let $j<m$ be the lowest outcome such that $t_j>0$. 
	Due to MLRP we can move weight from $t_j$ to $t_m$ at such a ratio that it will weakly decrease (resp., increase) the utility of all actions below (resp., above) $a_i$, and not change the utility of $a_i$. If at any point the utilities of $a_i$ and $a_k$ where $k>i$ become the same, then we have found a payment profile that incentivizes $a_k$ with the same utility to the agent as in the optimal contract that incentivizes $a_i$. But since the welfare of $a_k$ is larger, the principal's payoff must be larger in contradiction to the optimality of the contract. Thus we can move weight until $t_j$ becomes 0. We conclude that with MLRP and increasing welfare, there is always an optimal contract that pays only for the highest outcome.
\end{proof}

\section{Appendix on Properties and Geometry of Linearly-Implementable Actions}
\label{appx:geometry}

In this appendix we provide additional details for the structural insights developed in Section \ref{sec:geometry}.

\subsection{Basic Properties of Linearly-Implementable Actions}

We provide additional details for Observation~\ref{obs:single-implemented} and Observation~\ref{obs:indifferent}.

\begin{proof}[Proof of Observation~\ref{obs:single-implemented}]
	A contract implements more than one action only if two actions have the same maximum expected utility for the agent, and the same expected payoff for the principal (since the agent tie-breaks in favor of the principal). 
	In a linear contract with parameter $\alpha<1$, the principal's expected payoff is
	$(1-\alpha) R_{a}$, and as $R_a\ne R_{a'}$ for any two actions by Assumption A1, the necessary condition cannot hold.
	In a linear contract with parameter $\alpha=1$, the agent's expected utility is $R_a-c_a$, and so by Assumption A2 there is a unique action that maximizes this, completing the proof.
\end{proof}

\begin{proof}[Proof of Observation \ref{obs:indifferent}]
	The agent's expected utility from action $a$ is $\alpha_{a,a'}R_{a}-c_{a}$, which is equal by definition of $\alpha_{a,a'}$ to action $a'$'s utility $\alpha_{a,a'}R_{a'}-c_{a'}$. 
	For an example in which $\alpha_{a',a_2}$ incentivizes neither $a'$ nor $a_2$, consider the following. If the actions are $a_1=(F_1,c_1) = ((1, 0, 0), ~0)$, 
	$a'=(F',c') = ((0, 1, 0), ~1)$, and $a_2=(F_2,c_2) = ((0, 0, 1), ~2)$ and the outcomes are
	$(x_1,x_2,x_3)=(1,3,6)$, then only $a_1$ and $a_2$ can be implemented.
\end{proof}

\subsection{Structural Lemmas for Linear Implementability}

We provide proofs of Lemma \ref{lem:envelope-monotonicity} and Lemma \ref{lem:upper-envelope}. 

\begin{proof}[Proof of Lemma \ref{lem:envelope-monotonicity}]
	Notice that for every action $a$ and corresponding line $\alpha R_a-c_a$, the slope $R_a$ is non-negative. A key fact is that an upper envelope of affine functions with non-negative slopes is convex \cite{Rubinov08}. From convexity it follows that 
	the upper envelope crosses the $x$-axis at most once, so $u(\alpha)\in A_n\implies u(\alpha')\in A_n$.
	Also from convexity, the slopes of the line segments forming the upper envelope are increasing in $\alpha$, 
	so $R_a<R_{a'}$. By Assumption A2 it follows that $c_a<c_{a'}$. Assume for contradiction that $R_a-c_a> R_{a'}-c_{a'}$. But this means that segment $\ell_a$ intersects $\alpha=0$ at a higher point than segment $\ell_{a'}$ (as $-c_a>-c_{a'}$), and also intersects $\alpha=1$ at a higher point (as $R_a-c_a> R_{a'}-c_{a'}$). Segment $\ell_a$ thus completely overshadows $\ell_{a'}$, in contradiction to the fact that $u(\alpha')=a'$ and so $\ell_{a'}$ is part of the upper envelope. 
\end{proof}


\begin{proof}[Proof of Lemma \ref{lem:upper-envelope}]
	Fix $\alpha\in[0,1]$ and consider the linear contract with parameter $\alpha$. Action $a$ is IR when $\alpha R_a-c_a\ge 0$, i.e., if and only if its corresponding segment $\ell_a$ is at or above the $x$-axis at $\alpha$. Action $a$ is IC when $\alpha R_a-c_a\ge \alpha R_{a'}-c_{a'}$ for every $a'$, i.e., if and only if its segment $\ell_a$ participates in the upper envelope at $\alpha$. Thus both $a(\alpha)$ and $u(\alpha)$ return $\varnothing$ when all segments at $\alpha$ are below the $x$-axis, equiv., no action is IR given the linear contract with parameter $\alpha$. Otherwise, both return the action $a$ whose segment $\ell_a$ forms the upper envelope at $\alpha$ above the $x$-axis, equiv., the IC and IR action given the linear contract with parameter $\alpha$. In case of a tie, both break the tie in favor of the action with the highest expected outcome $R_a$---mapping $u(\cdot)$ does so by definition and mapping $a(\cdot)$ since this is the action that maximizes the principal's expected payoff $(1-\alpha)R_a$. This completes the proof.
\end{proof}

\subsection{Implications of Structural Lemmas}

We conclude with additional details for Corollary \ref{cor:what-a-returns}, Corollary~\ref{cor:increasing-welfare}, Corollary~\ref{cor:alphas-equal-intersections}, and Corollary \ref{cor:linear-equivalence}.

\begin{proof}[Proof of Corollary \ref{cor:what-a-returns}]
	Recall that mapping $a(\cdot)$ is onto $I_N$. Lemma \ref{lem:upper-envelope} shows that $a(\cdot)$ is equivalent to the upper envelope mapping $u(\cdot)$. 
	In the upper envelope, every segment appears once.
	This means that $a(\cdot)$ maps to action $a_i\in I_N$ for a consecutive range of $\alpha$s, starting at $\alpha_i$ (by its definition as the smallest $\alpha$ such that $a(\alpha)=a_i$). In the upper envelope,
	the segments $\{\ell_a\}$ are ordered by their expected outcomes $\{R_a\}$ (Lemma \ref{lem:envelope-monotonicity}).
	Since the actions in $I_N$ are also ordered by their expected outcomes (i.e., $R_1<\dots<R_N$), the range of $\alpha$s mapping to $a_i$ is immediately followed by the range mapping to $a_{i+1}$, establishing \eqref{eq:up-to-max}. The final range of $\alpha$s ending at $1$ maps to action $a_N$, completing the proof. 
\end{proof}

\begin{proof}[Proof of Corollary~\ref{cor:increasing-welfare}]
	Follows directly from Corollary \ref{cor:what-a-returns} and Lemma \ref{lem:envelope-monotonicity}. 
\end{proof}

\begin{proof}[Proof of Corollary~\ref{cor:alphas-equal-intersections}]
	For every $i\in[N]$, denote by $\ell_i$ the segment corresponding to action $a_i$.
	By Lemma~\ref{lem:upper-envelope} and Corollary~\ref{cor:what-a-returns}, parameter $\alpha_i$ is precisely the intersection point between $\ell_{i-1}$ and $\ell_i$ for every $i\ge 2$.
	Observe that the intersection between $\alpha R_i-c_i$ and $\alpha R_{i-1}-c_{i-1}$ is at point $\alpha_{i-1,i}$. It remains to consider the case of $i=1$, and in this case $\alpha_1$ is the intersection point between $\ell_1$ and the $x$-axis, which occurs at $\alpha=c_i/R_i$.  
\end{proof}

\begin{proof}[Proof of Corollary \ref{cor:linear-equivalence}]
	The first part of the corollary follows from Lemma \ref{lem:upper-envelope}, which establishes equivalence between the linearly-implementability mapping and the upper envelope mapping, and from the fact that the upper envelope depends only on lines $\alpha R_a-c_a$ parameterized by $R_a,c_a$. 
	The second part of the corollary follows from the fact that the principal's expected payoffs are $(1-\alpha)R_{a(\alpha)}$ and $(1-\alpha)R_{a'(\alpha)}$, respectively, and we have that $R_{a'(\alpha)}=R_{b(a(\alpha))}=R_{a(\alpha)}$. 
\end{proof}

\section{Appendix on Robust Optimality of Linear Contracts}
\label{appx:robustness}

In this appendix we provide the proofs that we omitted from Section~\ref{sec:robustness}.


\begin{proof}[Proof of Claim \ref{cla:downward-slope}]	
	In this case, let $F_i=F'_i$ for every action~$a_i$. We argue that the linear contract with parameter $\alpha=0$ has expected payoff at least as high as that of contract $t$. Observe that since $t_1>t_m$, the expected payments for the actions are decreasing in the actions' expected outcomes: if $R_i<R_k$ then $F_{i,m}=F'_{i,m}<F'_{k,m}=F_{k,m}$ and so $F_{i,1}>F_{k,1}$; thus $F_{i,1}t_1+F_{i,m}t_m > F_{k,1}t_1+F_{k,m}t_m$. Consider the zero-cost action $a_1$ (which exists by Assumption~A3), and let $a_{i^*}$ be the action incentivized by contract $t$. The expected outcome of action $a_{i^*}$ must be (weakly) lower than that of action $a_1$---its cost is (weakly) higher so its expected payment must be (weakly) higher. Since the agent's choice of action $a_{i^*}$ is IR, its expected outcome is an upper bound on contract $t$'s expected payoff to the principal. But the linear contract with parameter $\alpha=0$ incentivizes an action with (weakly) higher expected outcome at no cost to the principal, thus guaranteeing (weakly) higher expected payoff to the principal. This completes the proof for the case of $t_1>t_m$.
\end{proof}

\begin{proof}[Proof of Claim \ref{cla:non-affine}]
	A characterization of affine mappings is that they map every $3$ collinear points to points that are themselves collinear. Thus there must exist $3$ points $(x,t(x)),(x',t(x'))$ and $(x'',t(x''))$ where (w.l.o.g.) $x<x'<x''$, such that these points are not collinear. Now consider the line between the 2 points $(x_1,t_1)$ and $(x_m,t_m)$. It cannot be the case that the $3$ points $(x,t(x)),(x',t(x')),(x'',t(x''))$ are all on this line. Thus we have shown the existence of an index $j$ as required.
\end{proof}

\begin{proof}[Proof of Claim \ref{cla:below}]
	We prove the claim by showing that the affine contract with parameters $\alpha_0,\alpha_1$ (corresponding to $l_1$) has as high expected payoff to the principal as that of contract $t$ for distributions $\{F_i\}_{i=1}^n$, which are defined as follows:
	Let $a_{i^*}$ be the action incentivized by the affine contract with parameters $\alpha_0,\alpha_1$. Set $F_{i^*}=F'_{i^*}$ and $F_{k}=F''_{k}$ for every $k\ne i^*$. 	
	A similar argument as in the proof of Claim \ref{cla:above} establishes that for distributions $\{F_i\}_{i=1}^n$ as defined, contract $t$ will also incentivize action $a_{i^*}$, and at the same expected payment to the agent as the affine contract. This is because the expected payment for $a_{i^*}$ is the same in contract $t$ as in the affine contract, while the expected payments for all other actions (weakly) decrease compared to the affine contract.
\end{proof}

\begin{proof}[Proof of Observation \ref{obs:affine-vs-linear}]
	Fix $\alpha_0\ge 0$, and consider the mapping from $\alpha_1$ to the action implemented by the affine contract with parameters $\alpha_0,\alpha_1$. This mapping is identical to the mapping from $\alpha$ to the action implemented by the linear contract with parameter $\alpha=\alpha_1$. This follows from the analysis in Section \ref{sec:geometry} and in particular Lemma \ref{lem:upper-envelope},%
	\footnote{Note that we are not limiting $\alpha_1$ to be $\le 1$; this is not an issue since everything in Section \ref{sec:geometry} technically holds for $\alpha_1>1$ (of course, $\alpha_1>1$ does not  make sense for the principal since it leaves her with negative expected payoff).}
	and by observing that the line segments forming the upper envelope for the affine contract are the same as those of the linear contract, with an additional vertical shift of magnitude $\alpha_0$. So for every $\alpha_1$, the linear contract with $\alpha=\alpha_1$ implements the same action as the affine contract; but its expected payment to the agent is lower by $\alpha_0\ge 0$ than that of the affine contract. Thus its expected payoff to the principal is higher by $\alpha_0\ge 0$, completing the proof.
\end{proof}

\section{Appendix on Approximation Guarantees of Linear Contracts}
\label{appx:approximation}


In this appendix we provide additional details for the approximation guarantees developed in Section~\ref{sec:linear-contracts}.

\subsection{Auxiliary Lemma used in the Proof of Theorem \ref{thm:lower-bound}}

	\begin{lemma}\label{lem:aux-implementable}
	Consider the principal-agent settings $(A^\epsilon_n, \Omega_n)$ defined in the proof of Theorem~\ref{thm:lower-bound}. Then, $\alpha_i=\alpha_{i-1,i}$.
	\end{lemma}
	\begin{proof}
	Recall from Lemma \ref{lem:upper-envelope} that a linear contract with parameter $\alpha$ implements the action whose segment forms the upper envelope at $\alpha$. The next claim shows that
	for every $\alpha>\alpha_{i-1,i}$, the segment of every action $a_{i'}$ such that $i'<i$ is (weakly) below that of action~$a_{i}$. 
	It can similarly be shown that for every $\alpha<\alpha_{i,i+1}$ (or $\alpha\le 1$ for $i=n$), the segment of every action $a_{i'}$ such that $i'>i$ is below that of action $a_{i}$. 
	
	\begin{claim}
		\label{cla:alpha-i}
		For every $\alpha\ge \alpha_{i-1,i}$, the agent's expected utility from action $a_{i}$ is at least his expected utility from any ``previous'' action $a_{i'}$ where $i'<i$.
	\end{claim} 
	
	\begin{myproof}[Proof of Claim \ref{cla:alpha-i}]
		The claim holds trivially for the base case $i=1$.
		Assuming the claim holds for $i-1$, to establish it for $i$ it is sufficient to show that for every $\alpha\ge \alpha_{i-1,i}$, the agent's expected utility from action $a_{i}$ is at least his expected utility from action $a_{i-1}$.
		For every $\alpha=1-\epsilon^{i-1}+\delta$ where $\delta\ge0$, action $a_{i}$ has expected utility $\alpha R_{i} - c_{i} = (i-1)(1-\epsilon) + \frac{\delta}{\epsilon^{i-1}}$ for the agent, whereas action $a_{i-1}$ has expected utility $\alpha R_{i-1} - c_{i-1} = (i-1)(1-\epsilon) + \frac{\delta}{\epsilon^{i-2}}$, which is lower as required to establish Claim \ref{cla:alpha-i}. 
	\end{myproof}

	We conclude that for every $\alpha\in[\alpha_{i-1,i}, \alpha_{i,i+1})$ (or $\alpha\in[\alpha_{n-1,n}, 1)$ for $i=n$), the linear contract with parameter $\alpha$ implements action $a_i$, and so $\alpha_i=\alpha_{i-1,i}$. This completes the proof of Lemma \ref{lem:aux-implementable}.
\end{proof}

\subsection{Stronger Version of Theorem~\ref{thm:lower-bound}}\label{appx:LB-strong}

We show how to strengthen Theorem~\ref{thm:lower-bound} so that it requires only three outcomes. The proof of this lower bound also demonstrates that many actions can be linearly-implementable even if there are just a few outcomes.

\begin{theorem}
	\label{thm:strong-LB}
	For every number of actions $n$, there is a principal-agent setting $(A''_{n},\Omega_3)$ with $m = 3$ outcomes for which MLRP holds, 	
	such that the multiplicative loss in the principal's expected payoff from using a linear contract rather than an arbitrary one is at least $n$.
\end{theorem}

\begin{proof}
	Let $(A^{\epsilon}_n,\Omega_m)$ be the principal-agent setting defined in the proof of Theorem \ref{thm:lower-bound}, where $R_i=1/\epsilon^{i-1}$ and $c_i=R_i-i+\epsilon(i-1)$.
	Let $ALG$ be the best achievable expected payoff to the principal using a linear contract, and let $OPT$ be the same using an arbitrary contract. From the proof of Theorem~\ref{thm:lower-bound} we know that $ALG=1$ and $OPT\to n$ as $\epsilon\to 0$. To prove Proposition~\ref{thm:strong-LB}, we define two additional principal-agent settings, $(A'_n,\Omega_2)$ and $(A''_n,\Omega_3)$, with $ALG',OPT'$ and $ALG'',OPT''$ denoting their optimal linear and optimal expected payoffs, respectively. Our settings will be such that MLRP holds, and $ALG''\to 1$ while $OPT''\to n$, thus establishing the proposition.  
	
	%
	\paragraph{Auxiliary setting $(A'_n,\Omega_2)$.}
	Let $\Omega_2$ be an outcome space with $2$ outcomes $x_1=0,x_2=R_n$.
	For every action $a_i\in A^\epsilon_n$, we define a corresponding action $b(a_i)$ with the same cost $c_i$, which leads to outcome $x_2$ with probability $R_i/R_n$ and to outcome $x_1$ otherwise. The expected outcome of action $b(a_i)$ is thus $R_i$. 
	Let $A'_n=\{b(a_i)\}_{i\in[n]}$ be the collection of all actions corresponding to those in $A^\epsilon_n$. Observe that MLRP holds for $A'_n$, since for every $i'>i$, the likelihood ratio of outcome $x_2$ is $R_{i'}/R_i>1$, and the likelihood ratio of outcome $x_1$ is $(R_n-R_{i'})/(R_n-R_i)<1$. 
	
	Invoking Corollary~\ref{cor:linear-equivalence} for principal-agent settings $(A^{\epsilon}_n,\Omega_m)$ and $(A'_n,\Omega_2)$, we get that the principal's expected payoff from any linear contract is the same in both settings, so $ALG'=ALG= 1$. While action $b(a_n)$ has welfare approaching $n$ as $\epsilon\to0$, we are not in a full information setting and thus $OPT'$ may be much lower.
	
	\paragraph{Setting $(A''_n,\Omega_3)$.} Our goal now is to define a principal-agent setting $(A''_n,\Omega_3)$ for which $ALG''\approx ALG'$, $OPT''\approx n$, and MLRP still holds. We start from $(A'_n,\Omega_2)$ and add an outcome $x_3 = x_2+1=R_n+1$ to get the new outcome set $\Omega_3$. We change action $b(a_n)$ such that it leads to outcome $x_3$ with some small probability~$\delta$ (to be determined below); the probabilities over the other outcomes are renormalized by factor $(1-\delta)$. We denote the resulting action by $a''_n$, and its expected outcome by $R''_n=(1-\delta)R_n+\delta x_3=R_n+\delta$. We change every other action $b(a_i)$ only by adding zero probability that it leads to outcome $x_3$, and denote the new action by $a''_i$. The new action set $A''$ is $\{a''_i\}_{i\in[n]}$.
	Observe that MLRP still holds, since the likelihood ratio of action $a''_n$ and any other action $a''_i$ for outcome $x_3$ is $\infty$.
	
	In the new setting, $OPT''= R''_n - c_n=R_n+\delta-c_n=OPT+\delta$, by paying $c_n/\delta$ for outcome $x_3$ and zero for any other outcome, thus incentivizing the agent to choose action $a''_n$ while paying $c_n$ in expectation. As for $ALG''$, the only change relative to the original setting $(A^\epsilon_n,\Omega_m)$ and the auxiliary setting $(A'_n,\Omega_2)$ is that action $a''_n$ becomes linearly-implementable by a contract with a smaller parameter $\alpha$ than the original action $a_n$. Denoting this parameter by $\alpha''_n$, we have that $ALG''=(1-\alpha''_n)R''_n$, since linearly-implementing any other action has inferior expected payoff of~1. 
	By Corollary \ref{cor:alphas-equal-intersections}, $\alpha''_n=(c_n-c_{n-1})/(R''_n-R_{n-1})$.
	So
	\begin{eqnarray*}
		ALG'' &=& (1-\alpha''_n)R''_n\\
		&=& (1-\frac{c_n-c_{n-1}}{R_n+\delta-R_{n-1}})(R_n+\delta)\\
		&=&\frac{\epsilon^{n-1}-\epsilon^n+\delta\epsilon^{n-1}}{1-\epsilon+\delta\epsilon^{n-1}}(\epsilon^{-(n-1)}+\delta)\\
		&\le& (\epsilon^{n-1}+\delta\epsilon^{n-1})(\epsilon^{-(n-1)}+\delta)\\
		&=& 1 + \delta + \epsilon^{n-1}\delta(1+\delta).
	\end{eqnarray*}
	By letting $\delta\to0$ and $\epsilon\to0$ we get $OPT''/ALG''\to n$, completing the proof.
\end{proof}

\subsection{Proof of Theorem \ref{thm:upper-bound-C} and Corollary \ref{cor:lower-bound-C}}

\begin{proof}[Proof of Theorem \ref{thm:upper-bound-C}]
Consider the set $I_N$ of linearly implementable actions. Denote by $C = \max_{a \in I_N} c_a$ the highest cost of any of the linearly implementable actions. We will bucket the set of linearly implementable actions into $L = \lceil \log_2(C) \rceil$ buckets $B_1, \dots, B_L$ such that 
\[
B_i = \{a \mid 2^{i-1} \leq c_a < 2^i\}.
\]

Note that this bucketing ensures that every implementable actions is in some bucket. By Lemma \ref{lem:envelope-monotonicity} within each bucket actions are sorted simultaneously by expected outcome, cost, and welfare.

As in the proof of Theorem \ref{thm:upper-bound-H} let $h(k)$ denote the action $a \in B_k$ with the highest $R_a$ (and hence highest $c_a$), and let $l(k)$ denote the action $a \in B_k$ with the lowest $R_a$ (and hence lowest $c_a$). 
Now by the same argument as in Theorem \ref{thm:upper-bound-H}
\begin{align}
OPT \leq R_{h(k)} - c_{h(k)} \leq \sum_{k \leq L} (1-\alpha_{h(k),h(k)-1})R_{h(k)}. \label{eq:upx}
\end{align}

\begin{claim}
\label{cla:C1-or-C2}
For every bucket $B_k$ either (C1) $R_{l(k)} \geq R_{h(k)}/4$ or (C2) there exist an action $a_i \in B_k$ such that $R_{i} \geq R_{h(k)}/2$ and $\alpha_{i} \leq 1/2$.
\end{claim}

\begin{myproof}[Proof of Claim \ref{cla:C1-or-C2}]

Fix a bucket $B_k$. Consider the actions $l(k), \dots, h(k)$ in $B_k$.  Note that $\alpha_{l(k)} \leq \alpha_{l(k)+1} \leq \dots \leq \alpha_{h(k)}$. If Condition (C1) is met, we are done. So assume Condition (C1) is \emph{not} met. That is, assume that $R_{l(k)} < R_{h(k)}/4$. Note that this is possible only if $l(k) \neq h(k)$ and, thus, $l(k) < h(k)$. Further note that if $\alpha_{h(k)} \leq 1/2$ then Condition (C2) would be met by $i = h(k)$. So the only cases left are those where for a non-empty suffix of the indices $l(k), \dots, h(k)$ it holds that $\alpha_j > 1/2$.

We claim that it can't be that $\alpha_j > 1/2$ for all $j \in \{l(k), \dots, h(k)\}$. Indeed, if this was the case, we could use that for all $j' = l(k) + 1, \dots, h(k)$ by the definition of $\alpha_{j'}$
\[
{R_{j'}-R_{j'-1}} = \frac{1}{\alpha_{j'}} (c_{j'}-c_{j'-1}) \leq 2 \cdot (c_{j'}-c_{j'-1}).
\]
Summing this inequality over all $j' = l(k) + 1, \dots, h(k)$ would give us
\begin{align*}
R_{h(k)} - R_{l(k)} = \sum_{j' = l(k)+1}^{h(k)} (R_{j'}-R_{j'-1}) \leq \sum_{j' = l(k)+1}^{h(k)} 2 \cdot (c_{j'}-c_{j'-1}) = 2 \cdot(c_{h(k)} - c_{l(k)}).
\end{align*}
Since $c_{h(k)} < 2 c_{l(k)}$ and $c_{l(k)} \leq R_{l(k)}$ this would show
\[
R_{l(k)} \geq \frac{1}{3} \cdot R_{h(k)},
\]
but this would contradict our assumption that Condition (C1) is not met.

So there must be a largest index $i$ with $l(k) \leq i < h(k)$ for which it holds that $a_i \leq 1/2$ and $\alpha_{i'} > 1/2$ for all $i' > i$. We claim that this $i$ satisfies Condition (C2). It certainly has $a_i \leq 1/2$. For the expected outcome we can use the same argument that we used when we assumed that all the $\alpha_j$'s are strictly positive to conclude that
\[
R_{h(k)} - R_i = R_{h(k)} - R_{i'-1} \leq 2 \cdot (c_{h(k)} - c_{i'-1}) = 2 \cdot (c_{h(k)} - c_{i}).
\]
Because the actions in each bucket are sorted by costs, $c_i \geq c_{l(k)}$. Also, as we have argued before, $c_{h(k)} < 2 c_{l(k)}$ and $c_{l(k)} \leq R_{l(k)}$. So,
\[
R_{h(k)} - R_i \leq 2 \cdot (c_{h(k)} - c_{i}) \leq 2 c_{l(k)} \leq 2 R_{l(k)}.
\]
But now because Condition (C1) is not met
\[
R_{h(k)} - R_i \leq 2 R_{l(k)} \leq \frac{1}{2} R_{h(k)},
\]
which shows that $R_i \geq R_{h(k)}/2$ as claimed. This established Claim \ref{cla:C1-or-C2}.
\end{myproof}

The high level idea now is that in each bucket $B_k$ we will identify a linearly implementable action $\tau(k)$ whose expected payoff to the principal is at least one quarter of that bucket's contribution to the sum in on the RHS of inequality (\ref{eq:upx}).

Consider a fixed bucket $B_k$. We say that bucket $B_k$ is of Type 1 if it meets condition (C1) and of Type 2 if it meets condition (C2). For Type 1 buckets we choose $\tau(k) = l(k)$, and for Type 2 buckets we choose $\tau(k) = i$.

Then for Type 1 buckets:
\[
(1-\alpha_{h(k),h(k)-1}) R_{h(k)} \leq (1-\alpha_{l(k)}) R_{h(k)} \leq 4 \cdot (1-\alpha_{l(k)}) R_{l(k)} = 4 (1-\alpha_{\tau(k)}) R_{\tau(k)}
\]

And for Type 2 buckets:
\[
(1-\alpha_{h(k),h(k)-1}) R_{h(k)} \leq R_{h(k)} \leq 4 (1-\alpha_i) R_{i} = 4 (1-\alpha_{\tau(k)}) R_{\tau(k)}
\]

Where for the derivation of the inequalities for Type 1 buckets we used Claim 1 in the proof of Theorem \ref{thm:upper-bound-H}, which shows that $\alpha_{h(k),h(k)-1} \geq \alpha_{l(k)}$, and for Type 2 buckets we used that $\alpha_{h(k),h(k)-1} \geq 0$.

We thus get,
\begin{align*}
OPT &\leq \sum_{k \in [L]} (1-\alpha_{h(k),h(k)-1})R_{h(k)}\\
	&\leq L \cdot \max_{k \in [L]} \{(1-\alpha_{h(k),h(k)-1})R_{h(k)}\}\\
	&\leq 4 L \cdot \max_{k \in [L]} \{(1-\alpha_{\tau(k)}) R_{\tau(k)}\}\\
	&\leq 4 L \cdot ALG. \qedhere
\end{align*}
\end{proof}

\begin{proof}[Proof of Corollary \ref{cor:lower-bound-C}]
	Consider the construction of the lower bound in Theorem 3.3 (Workshop version) with $\epsilon = 1/2$. Then $c_i = \frac{1}{2}(2^i-i-1)$ for all $1 \leq i \leq n$. In particular, $c_{n} \in  (\frac{1}{2}\cdot2^{n-1},\frac{1}{2}\cdot2^n) = (\frac{C}{2},C)$ where we used $\log(2C) \geq 3$. With an arbitrary contract the principal can guarantee himself an expected payoff of $OPT = n-\epsilon(n-1) > \frac{n}{2}$, while with a linear contract he can achieve at most $ALG \leq 1.$ So $OPT/ALG \geq \frac{n}{2} = \Omega(\log C)$.
\end{proof}

\section{Appendix on Beyond Linear Contracts}
\label{appx:LB-monotone}

The instance we construct to prove Theorem \ref{thm:mon-lower-bound} is based upon the construction of Theorem~\ref{thm:strong-LB} for $n-1$ actions, with an additional $n$th high-cost action. Intuitively, the IC constraint with this extra action together with monotonicity enforce relatively homogeneous payments over the highest outcomes, whereas the optimal contract requires a single high payment for the second-highest outcome.

\begin{proof}[Proof of Theorem \ref{thm:mon-lower-bound}]

Let $\epsilon,\delta,\gamma>0$ be vanishingly small, and define an $m$-outcome vector and $n\times m$ distribution matrix as follows, where $m=4$:
\begin{eqnarray*}
	x &=& 
	\begin{pmatrix}
		0 & \frac{1}{\epsilon^{n-2}} & \frac{1}{\epsilon^{n-2}} + \gamma & \frac{1}{\epsilon^{n-2}} + 2\gamma
	\end{pmatrix},\\
	F &=& 
	\begin{pmatrix}
		1-\epsilon^{n-2} & \epsilon^{n-2} & 0 & 0\\			1-\epsilon^{n-1} & \epsilon^{n-1} & 0 & 0\\
		\vdots & \vdots & \vdots & \vdots \\
		1-\epsilon^{n-i-1} & \epsilon^{n-i-1} & 0 & 0\\		
		\vdots & \vdots & \vdots & \vdots \\
		1-\epsilon & \epsilon & 0 & 0\\
		0 & 1-\delta & \delta & 0\\
		0 & 0 & 0 & 1
	\end{pmatrix}.
\end{eqnarray*}
The costs are $c_i=\frac{1}{\epsilon^{i-1}} - i + \epsilon(i-1)$ for $i\le n-1$, and $c_n=\frac{1}{\epsilon^{n-2}}$.

Observe that in the above setting, MLRP holds; 
the expected outcomes are $R_i=\frac{1}{\epsilon^{i-1}}$ for $i\le n-2$,
$R_{n-1}=\frac{1}{\epsilon^{n-2}}+\delta\gamma$, and
$R_n=\frac{1}{\epsilon^{n-2}}+2\gamma$;
the expected welfares are $R_i-c_i=i-\epsilon(i-1)$ for $i\le n-2$, $R_{n-1}-c_{n-1}=n-1-\epsilon(n-2)+\delta\gamma$, 
and $R_n-c_n=2\gamma$. We now analyze several possible contracts to establish the theorem.

\paragraph{Optimal contract.}

The optimal contract incentivizes action $n-1$ by setting $t_3=\frac{c_{n-1}}{\delta}$ (such that actions $a_1,a_{n-1}$ both have expected utility 0 for the agent, all other actions have negative expected utilities, and tie-breaking is in favor of the principal). The expected payoff to the principal is $R_{n-1}-\delta t_3\approx n-1$ (action $a_{n-1}$'s welfare). 

\paragraph{Optimal monotone contract incentivizing $a_i\ne a_{n-1}$.}

We claim that the payoff that the principal can achieve by incentivizing any action $a_i$ with $i \neq n-1$ is at most $1$, so that it suffices to consider action $a_{n-1}$.
For action $a_1$ the payoff to the principal is upper bounded by the welfare this action obtains, which is $1$. 
For every action $a_i$ where $2 \le i\le n-2$, dis-incentivizing deviation from $a_i$ to $a_{i-1}$ necessitates 
\begin{align*}
(1-\epsilon^{n-i-1})t_1 + \epsilon^{n-i-1}t_2 - c_i &\geq (1-\epsilon^{n-i}) t_1 + \epsilon^{n-i}t_2 - c_{i-1}\\
\Leftrightarrow \hspace{78pt} (1-\epsilon)\epsilon^{n-i-1}t_2 &\geq \big(c_i - c_{i-1}\big) + (1-\epsilon) \epsilon^{n-i-1}t_1\\
\Rightarrow \hspace{133pt} t_2 &\geq \frac{1}{(1-\epsilon)\epsilon^{n-i-1}}\big(c_i - c_{i-1}\big),
\end{align*}
where
\begin{align*}
c_i - c_{i-1} = \left(\frac{1}{\epsilon^{i-1}}-i+\epsilon(i-1)\right) - \left(\frac{1}{\epsilon^{i-2}-(i-1)+\epsilon(i-2)}\right) = (1-\epsilon) \left(\frac{1}{\epsilon^{i-1}}-1\right).
\end{align*}
We get that
\begin{align*}
t_2 \geq \frac{1}{\epsilon^{n-i-1}}\left(\frac{1}{\epsilon^{i-1}}-1\right),
\end{align*}
and the payoff to the principal is at most
\begin{align*}
\frac{1}{\epsilon^{i-1}} - \epsilon^{n-i-1}\frac{1}{\epsilon^{n-i-1}}\left(\frac{1}{\epsilon^{i-1}-1}\right) = 1.
\end{align*}
The expected welfare of $a_n$ is almost zero. 

\paragraph{Optimal monotone contract incentivizing $a_{n-1}$.}

It is w.l.o.g.~to set $t_1=0$. 
To dis-incentivize deviations from $a_{n-1}$ to $a_{n-2}$ and to $a_n$:
\begin{eqnarray}
(1 - \delta)t_2 + \delta t_3 - c_{n-1} &\ge& \epsilon t_2 - c_{n-2},\label{eq:IC-down}\\
(1-\delta)t_2+\delta t_3-c_{n-1} &\ge& t_4-c_n.\label{eq:IC-up}
\end{eqnarray}
From \eqref{eq:IC-up} and monotonicity we have that 
$(1-\delta)t_2+\delta t_3 \ge t_4 - (c_n-c_{n-1}) \ge t_3 - (c_n-c_{n-1})$. So: 
$$
t_2 \ge t_3 - \frac{c_n-c_{n-1}}{1-\delta}.
$$ 
%
Combining this with \eqref{eq:IC-down} we get $(1-\delta-\epsilon)t_2\ge c_{n-1} - c_{n-2} -\delta t_3 \ge c_{n-1} - c_{n-2} -\delta t_2 - \frac{\delta(c_n-c_{n-1})}{1-\delta}$, and after rearranging,
$$
(1-\epsilon)t_2 \ge c_{n-1} - c_{n-2} - \frac{\delta(c_n-c_{n-1})}{1-\delta}.
$$
Since $c_{n-1} - c_{n-2}=\frac{1-\epsilon}{\epsilon^{n-2}} - (1-\epsilon)$,
$$
t_2\ge \frac{1}{\epsilon^{n-2}} - 1 - \frac{\delta(c_n-c_{n-1})}{(1-\delta)(1-\epsilon)}.
$$
So the expected payment in the optimal monotone contract incentivizing action $a_{n-1}$ is at least $\frac{1}{\epsilon^{n-2}}-1$ minus a term that is vanishing with $\delta$, leaving an expected payoff of $\approx 1$ for the principal.
This completes the proof of Theorem \ref{thm:mon-lower-bound}.
\end{proof}


%
%

\end{document}